\documentclass[lettersize,journal,twoside]{IEEEtran}
\usepackage{amsmath,amssymb,amsfonts}
\allowdisplaybreaks[4]
\usepackage{array}
\usepackage[caption=false,font=normalsize,labelfont=sf,textfont=sf]{subfig}
\usepackage{textcomp}
\usepackage{stfloats}
\usepackage{url}
\usepackage{verbatim}
\usepackage{graphicx}
\usepackage{cite}
\usepackage{subfig}
\usepackage[linesnumbered, ruled]{algorithm2e}
\SetKwRepeat{Do}{do}{while}%
\usepackage[thmmarks,amsmath]{ntheorem}
\theoremseparator{:}

\newtheorem{lemma}{Lemma}
\newtheorem{proposition}{Proposition}
\theoremheaderfont{\it}
\theorembodyfont{\normalfont}
\theoremsymbol{\ensuremath{\Box}}
\newtheorem*{proof}{Proof}
\theoremheaderfont{\it}
\theorembodyfont{\normalfont}
\theoremsymbol{}
\newtheorem{remark}{Remark}
\usepackage{setspace}
\begin{document}

\title{Edge Information Hub: Orchestrating Satellites, UAVs, MEC, Sensing and Communications for 6G Closed-Loop Controls}

\author{Chengleyang Lei,~\IEEEmembership{Student Member,~IEEE}, Wei Feng,~\IEEEmembership{Senior Member,~IEEE}, Peng Wei,~\IEEEmembership{Member,~IEEE}, Yunfei Chen,~\IEEEmembership{Senior Member,~IEEE}, Ning Ge,~\IEEEmembership{Member,~IEEE}, and Shiwen Mao,~\IEEEmembership{Fellow,~IEEE}
	\thanks{This work was supported in part by the National Natural Science Foundation of China under Grant 62425110, Grant 62341110 and Grant U22A2002, in part by the National Key Research and Development Program of China under Grant 2020YFA0711301, in part by the Suzhou Science and Technology Project, and in part by the FAW Jiefang Automotive Co., Ltd.}
	\thanks{Chengleyang Lei, Wei Feng (corresponding author), and Ning Ge are with the Department of Electronic Engineering, State Key Laboratory of Space Network and Communications, Tsinghua University, Beijing 100084, China (email: lcly21@mails.tsinghua.edu.cn, fengwei@tsinghua.edu.cn, gening@tsinghua.edu.cn).}
	\thanks{Peng Wei is with the National Key Laboratory of Wireless Communications, University of Electronic Science and Technology of China, Chengdu 611731, China (e-mail: wppisces@uestc.edu.cn).}
	\thanks{Yunfei Chen is with the Department of Engineering, University of Durham, DH1 3LE Durham, U.K. (e-mail: yunfei.chen@durham.ac.uk).}
	\thanks{Shiwen Mao is with the Department of Electrical and Computer Engineering, Auburn University, Auburn, AL 36849, USA (e-mail: smao@ieee.org).}
}




\maketitle
\pagestyle{empty}
\thispagestyle{empty}

\begin{abstract}
An increasing number of field robots would be used for mission-critical tasks in remote or post-disaster areas. Due to the limited individual abilities, these robots usually require an edge information hub (EIH), with not only communication but also sensing and computing functions. Such EIH could be deployed on a flexibly-dispatched unmanned aerial vehicle (UAV). Different from traditional aerial base stations or mobile edge computing (MEC), the EIH would direct the operations of robots via sensing-communication-computing-control ($\textbf{SC}^3$) closed-loop orchestration. This paper aims to optimize the closed-loop control performance of multiple $\textbf{SC}^3$ loops, with constraints on satellite-backhaul rate, computing capability, and on-board energy. Specifically, the linear quadratic regulator (LQR) control cost is used to measure the closed-loop utility, and a sum LQR cost minimization problem is formulated to jointly optimize the splitting of sensor data and allocation of communication  and computing resources. We first derive the optimal splitting ratio of sensor data, and then recast the problem to a more tractable form. An iterative algorithm is finally proposed to provide a sub-optimal solution. Simulation results demonstrate the superiority of the proposed algorithm. We also uncover the influence of $\textbf{SC}^3$ parameters on closed-loop controls, highlighting more systematic understanding.
\end{abstract}


\section{Introduction}
Robots and various unmanned machines have great potential to help humans carry out dangerous and strenuous tasks in post-disaster or remote areas~\cite{robot1,robot2}. In general, the abilities of an individual robot are usually limited. For example, the sensors equipped on a robot could only detect its surrounding information. The onboard computers may malfunction due to the electronic components failure caused by the harsh conditions after disasters, including high temperatures and radiation~\cite{emergency1}. In such cases, the robots have to rely on external helpers to assist them with mission-critical control tasks. Such a helper should be integrated with sensing functionality for global environmental detection~\cite{sensor1}, computing functionality for sensor-data analysis and decision-making~\cite{computing1}, and communication functionality for delivering control commands~\cite{communication1}. 

Traditionally, the designs of the sensing, communication, and computing modules are relatively independent of each other, with little consideration of their relationships. However, during certain tasks, these components often collaborate closely to support robots. This suggests that the helper should be considered as a unified entity integrating sensing, communication, and computing functionalities. As an integrated center of the control-oriented information, we refer to the above helper as an edge information hub (EIH), which incorporates remote sensors, mobile edge computing (MEC) servers and communication modules. Note that, the terrestrial infrastructures may be unavailable during disasters. For more robust applications, the EIH can be deployed on a flexibly-dispatched unmanned aerial vehicle (UAV) in an on-demand manner~\cite{Feng2024_1}. In addition, the payload of an EIH is limited in practice, resulting in limited computing capability onboard. Accordingly, the satellite could be leveraged to offload some sensor data to the remote cloud. This leads to EIH-based satellite-UAV networks for control tasks. 

The EIH assists the field robots to accomplish control tasks in a closed-loop manner. Specifically, the remote sensors collect information of the controlled objects, and then the computing modules (comprising the MEC server and the remote cloud) process the sensor data for decision making and control commands. Next, the communication modules transmit the commands to the corresponding robots. Finally, the robots follow the commands to perform the tasks. This entire process is referred to as the sensing-communication-computing-control ($\textbf{SC}^3$) loops~\cite{wcl}. Different components of the $\textbf{SC}^3$ loops are coupled with each other. Therefore, different from the traditional system design focusing on communications only, it becomes more relevant to focus on the whole closed-loop performance and to design different components of the $\textbf{SC}^3$ loops with a systematic mindset. In addition, constrained by the payload of UAVs, both communication and computing resources on the EIH are limited~\cite{MEC1,MEC6,MEC7, Feng2024_2}. This makes it necessary to orchestrate the communication and computing resources to ensure that resource allocation aligns with control requirements. Motivated by these considerations, we investigate the joint communication and computing resource allocation problem in an EIH-empowered system serving multiple robots for their control tasks, with the aim of optimizing the closed-loop control performance.

\subsection{Related Works}
Closed-loop control utilizes the output of a dynamic system as the input of the controller, which forms a closed loop\cite{closed_loop}. It is an important research field in the control theory, as it can stabilize an unstable system and reduce sensitivity to disturbance\cite{closed_loop1}. Classical control theory, rooted in frequency domain techniques based on transfer functions, has been extensively investigated\cite{frequency_method1,frequency_method2}. In the late 1950s, researchers began to develop the modern control theory, employing the state variable approach\cite{modern_control1,modern_control2}. An important area within modern control theory is optimal control, which aims at seeking a control strategy that optimizes an objective function\cite{optimal_control1,optimal_control2}. In optimal control, the linear quadratic regulator (LQR) control cost is commonly used to measure the system state deviation and control input energy\cite{LQR1}. The optimal control strategy to minimize the LQR cost has been proven to be a linear strategy\cite{LQR2}. Nevertheless, most works in the control field assume that the communication limitations in the loops have a negligible impact on control performance\cite{control1}. However, the resources on the EIH are usually limited by the UAV payload and flight time. In such cases, it is necessary to take the limitations of communication and computing capability into consideration.

As an integration of computing and physical processes, cyber-physical systems (CPSs), which require computation, communication, control, and perception, have received great attention \cite{CPS1}. Some works have delved into communication design to meet high requirements in the CPSs~\cite{CPS2, CPS3, CPS4}. In particular, \cite{CPS2} proposed the design and implementation of a real-time high-speed wireless communication protocol called RT-WiFi, so as to provide deterministic timing at a high sampling speed in CPSs. The authors in \cite{CPS3} analyzed the options and configuration choices for the fifth generation (5G) network deployment in industrial CPSs. In \cite{CPS4}, an energy-efficient massive multiple-input multiple-output (MIMO) system was designed to reduce the power consumption and fabrication cost for the industrial CPSs. On the other hand, in order to effectively manage the diverse data with low latency, MEC has been investigated as a means of integration in the CPSs\cite{CPS5, CPS6}. The authors in \cite{CPS5} proposed a deep Q-network-based service placement algorithm, which minimizes the service response time by optimizing the service placement, resource allocation and workload scheduling. Ref. \cite{CPS6} investigated large-scale CPSs, and proposed an edge intelligent approach to minimize the service latency and maximize system lifetime. Although these works have provided valuable insights into $\textbf{SC}^3$ integration, further research is needed to comprehend the coupling relationships among the sensing, communication, computing, and control components in such systems. Exploring how the sensing, communication, and computing capabilities influence the overall control performance will be beneficial for a more effective $\textbf{SC}^3$ system design.

As a special type of CPSs, networked control systems (NCSs), wherein the control systems are connected through a communication network, have been recently studied due to their flexibility and maintainability~\cite{NCS,NCS1}. Many works have been conducted to investigate the impact of communication limitations on the control system stability from various aspects, including data rate~\cite{control1}, delay~\cite{NCS2}, packet loss~\cite{NCS3, NCS9}, and so on.  In \cite{control1}, the authors demonstrated that the control system can be stabilized only if the communication data rate exceeds its intrinsic entropy rate. Ref. \cite{NCS2} investigated an NCS consisting of clock-driven sensors and event-driven controllers and actuators, and analyzed the stability region plot with respect to the sampling rate and network-induced delay. Authors in \cite{NCS3} considered the logarithmic quantization and packet loss, and derived the stability condition. Ref. \cite{NCS9} further considered the packet loss and random delays in NCSs. A set of necessary and sufficient conditions for stabilizing the NCSs was proposed. On the other hand, some works designed control strategies contemplating communication limitations\cite{NCS4, Yue2013, NCS5,NCS8,NCS11}. A control strategy aiming at achieving good performance over an unreliable communication network affected by packet loss and delays was described in \cite{NCS4}, which uses the data packet frame to transmit control sequences. The authors in \cite{Yue2013} proposed an event-triggered control strategy that guaranteed stability with an $H_{\infty}$ norm bound, where the communication delay was considered. In \cite{NCS5}, the control strategy and transmit power policies were designed to minimize the weighted sum of the control cost and power cost, where the packet loss was considered. Ref. \cite{NCS8} investigated the tradeoff between the data rate and the control performance. The transmission and control strategies were proposed to minimize the sum of the communication cost and control cost function. Authors in \cite{NCS11} considered the packet loss and optimized control parameters under the system instability probability constraint. In addition, some works investigated the influence of communication indicators on control performance~\cite{LQR, NCS6,NCS12}. In \cite{LQR}, a lower bound of the minimum average data rate to achieve a certain LQR cost was presented. Authors in \cite{NCS6} investigated  connected and automated vehicles, where the impact of communication erasure channels on vehicle platoon formation and robustness was analyzed. Ref. \cite{NCS12} further investigated the tradeoff between the average data rate and control performance of NCSs, considering the transmission delay. All of these works are valuable for analyzing the $\textbf{SC}^3$ loops. However, most of these works only investigated one control loop from the perspectives of performance analysis or control strategy design, instead of resource allocation among multiple $\textbf{SC}^3$ loops.

Recently, some works have been conducted on the communication resource allocation among $\textbf{SC}^3$ loops considering the control performance. Constraints on the control performance have been considered in the resource allocation~\cite{NCS7,NCS10,NCS13}. The authors in \cite{NCS7}  studied the resource allocation problem in wireless control systems, where the bandwidth and transmit power were jointly optimized to maximize the spectral efficiency, under the control convergence rate constraint. Ref. \cite{NCS10} designed a frequency allocation policy to keep the overall control system stable. In \cite{NCS13}, the authors investigated a system identification problem in a wireless NCS. The transmit power and channel allocation were jointly optimized to maximize the communication throughput or minimize the power consumption, while guaranteeing the system identification performance. Unlike these works that focused on the communication process, in an $\textbf{SC}^3$ system, the communication process is only designed to facilitate control and it is the control performance that should be the main objective for optimization. Some studies investigated the communication resource allocation aiming at improving the control performance~\cite{NCS14,NCS15,wcl}. Ref. \cite{NCS14} maximized the ratio of the remained energy to LQR cost by optimizing the transmit power and time allocation. The authors in \cite{NCS15} extended the work in \cite{NCS14} by introducing reconfigurable intelligent surface (RIS) technology to assist communication in a wireless NCS, where the transmission power, time, beamforming matrix, and RIS reflecting coefficients were jointly optimized to maximize the ratio of the reliability performance to LQR cost. In our previous work \cite{wcl}, we formulated an LQR cost minimization problem that optimized the transmit power allocation. However, these works did not utilize MEC or consider the computing resource allocation to further reduce latency.

On the other hand, the UAV-aided MEC has been widely investigated as a potential technology to extend the coverage of computation service\cite{MEC3, MEC4, MEC8, MEC9}. The authors in \cite{MEC3} considered a UAV-aided MEC system, where the user data can be processed locally or offloaded to the MEC. The sum of delays was minimized by jointly optimizing the UAV trajectory, the ratio of offloading tasks, and the user scheduling variables. The authors in \cite{MEC4} further proposed a multi-agent reinforcement learning method to solve the resource allocation problem in an MEC- and UAV-assisted vehicular network. Ref. \cite{MEC8} considered the information security in the UAV-assisted MEC system, where the secure computation efficiency was maximized by optimizing the offloading decision and resource allocation based on deep reinforcement learning. The authors in \cite{MEC9} investigated a vehicular edge computing system, where a UAV was utilized to assist task offloading. A UAV-assisted vehicular task offloading problem was proposed to minimize vehicular task delay. Most existing works on the UAV-aided MEC focus on communication performance, such as latency or energy efficiency. However, as mentioned above, in $\textbf{SC}^3$ loops, our main concern is the overall closed-loop control performance, rather than the separate communication or computing performances. Therefore, the closed-loop performance is of more interest to the system design.

In summary, despite the above works, several outstanding challenges remain in the $\textbf{SC}^3$ system design. Firstly, the sensing, communication, computing, and control components in the $\textbf{SC}^3$ loops cooperate closely to accomplish common tasks, and they are coupled with each other. However, existing works mainly focus on the aforementioned four components separately, lacking a holistic consideration of the entire system. Secondly, to efficiently utilize the limited resources on the EIH and improve the overall control performance, the multi-domain resources, such as the transmit power, the communication rate, and the computing capability, should be jointly orchestrated. There is lack of work focusing on the control performance to concurrently address the joint allocation of multi-domain resources. Therefore, it is imperative to explore the joint allocation of the multi-domain resources with the objective of the control performance.

\subsection{Main Contributions}
Motivated by the above observations, in this paper, we investigate an EIH-empowered $\textbf{SC}^3$ system where an EIH assists multiple robots with their control tasks. The UAV-mounted EIH integrates the sensing, computing, and communication capabilities to assist the robots, and utilizes satellite as backhaul. We jointly optimize the communication and computing resources, as well as the splitting ratio of sensor data, to minimize the sum LQR cost of multiple loops. The optimization problem is a non-convex problem. We recast it to a more tractable form and propose an iterative algorithm to solve it. The main contributions are summarized as follows.
\begin{itemize}
	\item We conceive the sensing, communication, and computing modules as a unified entity that assists robots to accomplish tasks, referred to as an EIH. This allows us to focus on the coupling among the functions within the $\textbf{SC}^3$ loops. By comprehensively contemplating the influence of different components and orchestrating multi-domain resources, we can improve the overall closed-loop control performance efficiently.
	\item We investigate a UAV-mounted EIH, which is integrated with a remote sensor, an MEC server, and a communication module to assist multiple robots with their control tasks. The sensor data can be processed locally on EIH, offloaded to the cloud after pre-processing, or offloaded to the cloud without pre-processing. In order to explore the potential of closed-loop orchestration, we jointly optimize the splitting of sensor data and allocation of communication and computing resources. Specifically, we utilize the LQR cost to evaluate the overall control performance, and incorporate the minimum information entropy constraint to achieve a certain LQR cost.
	\item We formulate a sum LQR cost minimization problem, which jointly optimizes the splitting ratio of sensor data, the computing capability of MEC, the satellite-backhaul rate, and the transmit power from EIH to robots. The problem is non-convex.
	We derive the optimal splitting vector of sensor data, and accordingly recast the original problem to a more tractable form. Finally, we propose an iterative algorithm to solve the recast problem based on the successive convex approximation (SCA) method.
	\item We provide simulation results to show the superiority of the proposed closed-loop-oriented method over the traditional communication-oriented method. Additionally, we show how the sensing noise variance and the computing capability influence the LQR cost through simulation.
\end{itemize}

\subsection{Organization and Notation}
The rest of the paper is organized as follows. Section \ref{sec_system} introduces the system model, and formulates the optimization problem. In Section \ref{sec_algorithm}, we recast the original problem to a more tractable form and propose an iterative algorithm to solve it. Simulation results are provided in Section \ref{sec_simulation} with further discussions. Finally, Section \ref{sec_conclusion} concludes this paper.

Throughout this paper, lower case and upper case boldface symbols denote vectors and matrices, respectively. $\mathbb{R}^{n}$ represents the collection of all the $n$-dimensional real-number vectors, and $\mathbb{R}^{m\times n}$ represents the collection of all the $m\times n$ real-number matrices. $\mathbb{E}$ is the expectation operator. $\text{tr} \left(\cdot \right)$ and $\det \left(\cdot \right)$ denote the trace operator and the determinant operator, respectively.
      
\section{System Model and Problem Formulation}
\label{sec_system}
As shown in Fig. \ref{fig:system}, we consider an EIH-empowered $\textbf{SC}^3$ system, where $K$ field robots are assisted by the EIH to perform mission-critical tasks. The UAV-mounted EIH incorporates a remote sensor, an MEC server, and a communication module to synergistically integrate sensing, computing, and communication functions. Due to the limited computing capability on the EIH, part of the sensor data will be offloaded to the cloud server through the satellite.

The EIH directs the field robots by forming $\textbf{SC}^3$ closed loops. In each loop, the remote sensor captures the states of the controlled object. The sensor data is then processed in the MEC server or the cloud to judge the situation and make corresponding control commands. Next, the communication module sends these commands to the field robots. The field robots follow the received commands to handle their object. The whole process is performed periodically. In the following, we will detail the models of different parts of the $\textbf{SC}^3$ loop.

\begin{figure} [t]
	\centering
	\includegraphics[width=0.99\linewidth]{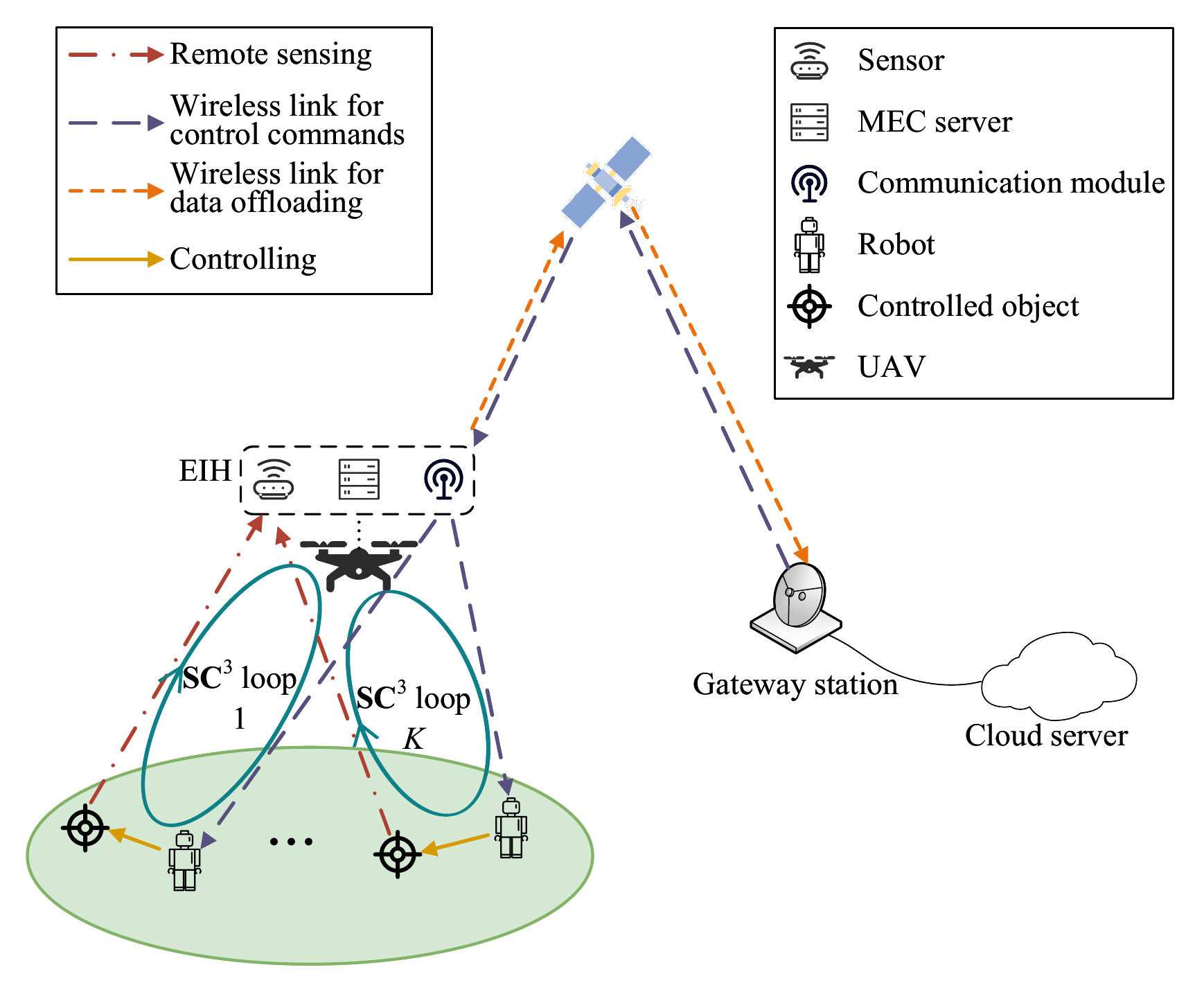}
	\caption{Illustration of an EIH-empowered $\textbf{SC}^3$ system, where the EIH is mounted on a UAV, and utilizes satellites to backhaul data.}
	\label{fig:system}
\end{figure}

\subsection{Computation Model}
In each cycle, the sensor data is analyzed to compute the optimal control commands. Due to the limited computing capability of the MEC server in the EIH, some of the sensor data will be offloaded to the cloud through satellite. As shown in Fig. \ref{fig:computing}, we assume the sensor data can be arbitrarily split into three parts:
\begin{itemize}
	\item Part 1: processed in the MEC server completely.
	\item Part 2: pre-processed in the MEC server and then transmitted to the cloud for further processing.
	\item Part 3: processed in the cloud completely.
\end{itemize}
The data sizes of the three parts of the sensor data in loop $k$ are denoted as $D_{k,1}$, $D_{k,2}$ and $D_{k,3}$ in bits. We have
\begin{equation}\label{Dk}
D_{k,1}+D_{k,2}+D_{k,3} = D_k,
\end{equation}
where $D_k$ denotes the total size of the sensor data of $\textbf{SC}^3$ loop $k$ in each cycle.

\begin{figure} [t]
	\centering
	\includegraphics[width=0.99\linewidth]{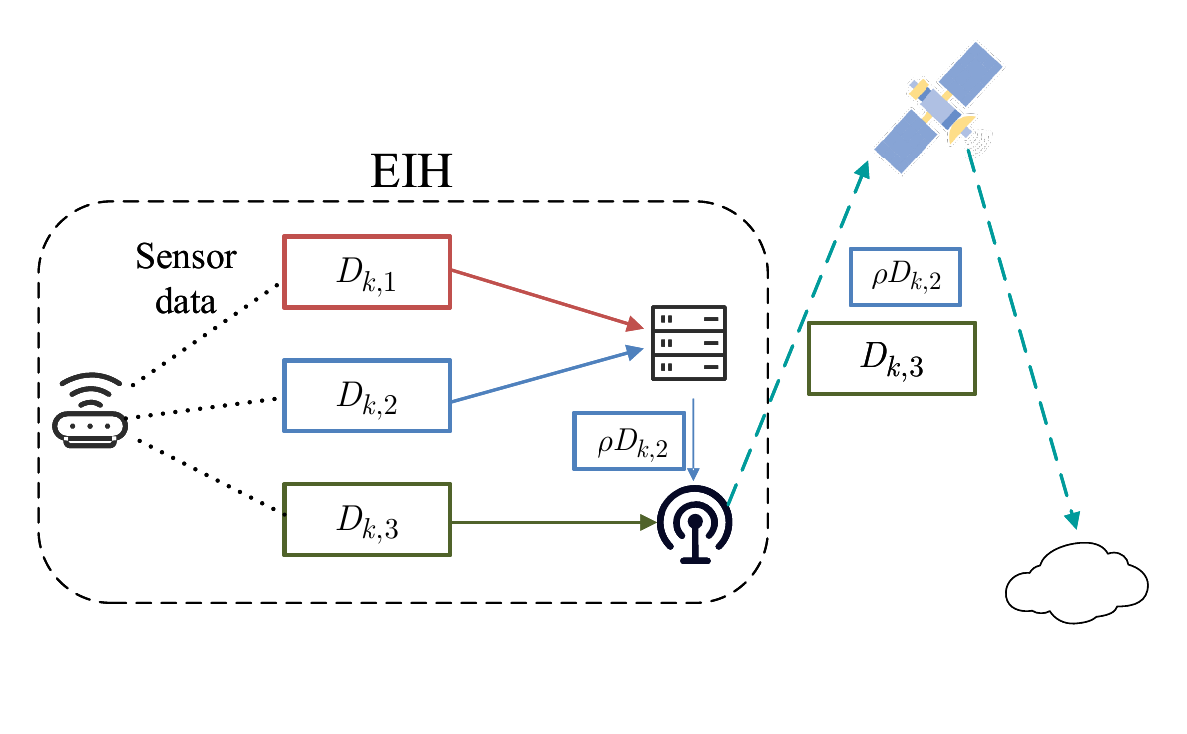}
	\caption{Illustration of three flows of the remote sensor data.}
	\label{fig:computing}
\end{figure}

The three parts of sensor data are processed in parallel as data streams. The overall computation time depends on the maximum processing time of the three parts.
 
For the first part of sensor data, the computation time can be formulated as
\begin{equation}\label{Tk1}
T_{k,1}^{\text{comp}} = \frac{\alpha D_{k,1}}{f_{k,1}},
\end{equation} 
where $\alpha$ denotes the number of CPU cycles for processing the sensor data per bit, and $f_{k,1}$ is the computing capability (i.e., CPU frequency) allocated to first part of data in loop $k$ by the MEC.

For the second part of sensor data, the data are first pre-processed in the MEC server, where the processing time can be calculated as
\begin{equation}\label{Tcompk2}
T_{k,2}^{\text{proc}} =
 \frac{\beta D_{k,2}}{f_{k,2}},
\end{equation} 
where $\beta$ denotes the number of CPU cycles for pre-processing sensor data per bit, and $f_{k,2}$ is the computing capability allocated to second part of data in loop $k$ by the MEC server.

We assume the data compression ratio of pre-processing is $\rho$, i.e., data of $\rho D_{k,2}$ bits will be transmitted to the satellite for further processing after pre-processing. The transmission latency from the UAV to satellite can be calculated as
\begin{equation}\label{Ttransk2}
T_{k,2}^{\text{trans}} = \frac{\rho D_{k,2}}{R_{k,2}},
\end{equation}
where $R_{k,2}$ is the backhaul rate from UAV to satellite allocated to the second part of data in loop $k$.

The downlink transmission data rate from the satellite to the cloud is usually much bigger than the uplink data rate. Therefore, the downlink transmission latency is negligible compared with $T_{k,2}^{\text{trans}}$. In addition, we assume that the cloud has enough computing capability, so that the computation time in the cloud is negligible. The transmission latency of the output data is also assumed to be ignorable as  the output data is much smaller than the input data size. Based on the above analysis, the overall time for processing the second part of sensor data can be calculated as
\begin{equation}\label{Tk2}
T_{k,2}^{\text{comp}} = \left\{
\begin{aligned}
&0, &\text{if}\  D_{k,2} = 0\\
&\max \left\{T_{k,2}^{\text{proc}}, T_{k,2}^{\text{trans}}\right\} + 4\tau, &\text{if}\  D_{k,2} > 0,
\end{aligned}
\right
.
\end{equation}
where $\tau$ is the propagation latency between the ground and the satellite. The pre-processing and the transmission process are performed in parallel. Therefore, the overall latency will be determined by the maximum of the two processes.  It should be noted that we have considered the special case of $D_{k,2} = 0$, i.e., we do not split any data to Part $2$. In such case, the latency for the second part of data is zero.

For the third part of data, similarly, the total time can be calculated as 
\begin{equation}\label{Tk3}
T_{k,3}^{\text{comp}}= \left\{
\begin{aligned}
&0, &\text{if}\  D_{k,3} = 0\\
&\frac{ D_{k,3}}{R_{k,3}} + 4\tau,\ &\text{if}\  D_{k,3} > 0,
\end{aligned}
\right
.
\end{equation}
where $R_{k,3}$ is the satellite-backhaul rate allocated to the third part of sensor data in loop $k$.

In conclusion, the overall time duration of the computing phase can be formulated as
\begin{equation}\label{Tcompk}
T_{k}^{\text{comp}} = \max \left\{T_{k,1}^{\text{comp}}, T_{k,2}^{\text{comp}}, T_{k,3}^{\text{comp}}\right\}.
\end{equation}

Considering the limited computation and communication resources, we have the following constraints
\begin{align}
\sum_{k = 1}^{K} \left( f_{k,1} + f_{k,2}\right)  \leq F_{\text{max}},
\end{align}
\begin{align}
\sum_{k = 1}^{K} \left( R_{k,2} + R_{k,3}\right)  \leq R_{\text{max}}^{\text{U2S}},
\end{align}
where $F_{\text{max}}$ denotes the computing capability of the MEC server, and $R_{\text{max}}^{\text{U2S}}$ denotes the maximum satellite-backhaul rate.

\subsection{Communication Model}
After processing the sensor data, the control commands will be transmitted to the field robots. The EIH transmits the control commands to $K$ robots simultaneously through orthogonal channels\footnote{In this paper, an assumption is made that different channels are completely orthogonal, to neglect the interference between the robots. This allows us to concentrate on the control-performance-oriented resource allocation and data offloading scheme. The inclusion of interference among robots will be a future study.}. Denoting the transmit power allocated to loop $k$ as $p_k$, we have 
\begin{equation}
\sum_{k = 1}^{K} p_k \leq P_{\text{max}},
\end{equation}
where $P_{\text{max}}$ represents the transmit power constraint.

The EIH communications with satellites and field robots in different frequency bands, so we assume that there is no communication interference among different communication links. The wireless channels between the EIH and robots are assumed to be dominated by line-of-sight (LoS) links~\cite{Zeng2016}. Therefore, the channel gain from the EIH to robot $k$ follows the free space path loss model as $g_k = \frac{\gamma_0}{d^2_k}$, where $d_k$ denotes the distance from the UAV to robot $k$ and $\gamma_0$ is the reference channel gain at the distance of one meter. The transmit data rate from EIH to robot $k$ can be calculated as\footnote{The small-scale channel fading is not considered here. In more diverse environmental conditions where the small-scale fading should be taken into account, we can regarded \eqref{eq12} as an upper bound for the	expected data rate}
\begin{equation}\label{eq12}
R_k^{\text{U2G}}\left( p_k \right)  = \log_2\left( 1 + \frac{g_kp_k}{\sigma^2}\right),
\end{equation}
where $\sigma^2$ denotes the channel noise power.

The computing phase and communication phase of each loop share the time resource, which can be denoted as
\begin{equation}
	T_{k}^{\text{comp}}+T_{k}^{\text{commu}}\le T_k,
\end{equation}
where $T_{k}^{\text{commu}}$ denotes the transmission time of the control commands, and $T_k$ is the time recourse reserved for the computing and communication phase, which is assumed to be fixed in each cycle.

\subsection{Control Model}
After receiving the control commands from the EIH, the robots follow the commands to handle the objects. For simplicity, we model each robot and its object as a linear control system\footnote{Although some control systems may be quite complicated, they can still be analyzed as linear systems through local linearization\cite{linearization}.}, and formulate the discrete-time system equation of the $k$-th control system in cycle $t$  as~\cite{LQR}
\begin{equation}\label{system}
	\mathbf{x}_{k,t+1} = \mathbf{A}_{k}\mathbf{x}_{k,t}+\mathbf{B}_{k}\mathbf{u}_{k,t}+\mathbf{v}_{k,t},
\end{equation}
where $t$ denotes the cycle index, $\mathbf{x}_{k,t}\in \mathbb{R}^{n_k}$ denotes the system state, such as the temperature or radiation intensity,  $\mathbf{u}_{k,t}\in \mathbb{R}^{m_k}$ denotes the control input, $n_k$ and $m_k$ denote the dimensions of the system state and control input, respectively, $\mathbf{v}_{k,t}\in \mathbb{R}^{n_k}$ denotes Gaussian control system noise with mean zero and covariance $\mathbf{\Sigma}^{\text{V}}_k$, and $\mathbf{A}_{k}$ and $\mathbf{B}_{k}$ are fixed $n_k\times n_k$ and $n_k\times m_k$ matrices denoting the state matrix and input matrix. The matrix $\mathbf{A}_{k}$ describes how a state depends on the previous state, and the matrix $\mathbf{B}_{k}$ describes how the control input affects the system state. The two matrices are determined by the type of control task, which are assumed to be known to the EIH.

We also consider a linear sensing model, where the  observation equation can be written as
\begin{equation}\label{sensor}
	\mathbf{y}_{k,t} = \mathbf{C}_k \mathbf{x}_{k,t} + \mathbf{w}_{k,t},
\end{equation}
where $\mathbf{y}_{k,t}\in \mathbb{R}^{q_k}$ is the sensing output, $\mathbf{C}_k\in \mathbb{R}^{q_k\times n_k}$ is the observation matrix that defines how the system state is transformed into an output vector, $q_k$ denotes the dimension of the sensing output, and $\mathbf{w}_{k,t}\in \mathbb{R}^{q_k} $ is the Gaussian sensing noise with mean zero and covariance $\mathbf{\Sigma}^{\text{W}}_k$.

In this paper, we evaluate the control performance with the infinite-horizon LQR cost\cite{LQR}, which is formulated as
\begin{equation}\label{LQR}
l_k \triangleq \lim\limits_{N\rightarrow \infty}\mathbb{E} \left[\frac{1}{N} \sum_{t = 1}^{N} \left(\mathbf{x}_{k,t}^\text{T}\mathbf{Q}_{k}\mathbf{x}_{k,t} +\mathbf{u}_{k,t}^\text{T}\mathbf{R}_{k}\mathbf{u}_{k,t}\right) \right],
\end{equation}
where $\mathbf{Q}_k$ and $\mathbf{R}_k$ are semi-positive definite weight matrices. The term $\mathbf{x}_{k,t}^\text{T}\mathbf{Q}_{k}\mathbf{x}_{k,t}$ denotes the deviation of the system from zero state, and the term $\mathbf{u}_{k,t}^\text{T}\mathbf{R}_{k}\mathbf{u}_{k,t}$ denotes the control energy consumption. The physical significance of the LQR cost function is that it is a comprehensive measure of the system state convergence and control energy consumption. A lower LQR cost indicates smaller system state deviation and less control energy consumption, which implies a better control performance. The weight matrices $\mathbf{Q}_k$ and $\mathbf{R}_k$ balance the state and the energy, which can be set according to the practical requirements. For example, one should set the entries of $\mathbf{Q}_{k}$ to be large if one wants the state of the system to converge to zero quickly but is not so concerned about energy consumption.

As in \eqref{LQR}, the LQR cost is determined by the system state $\mathbf{x}_{k,t}$ and the control input $\mathbf{u}_{k,t}$, which is mainly influenced by the control strategy that computes the control input based on the sensing output. We will not try to optimize the control strategy, which is beyond the scope of this work. Instead, we will focus on the impact of the communication capability on the LQR cost, and accordingly carry out computing and communication resource allocation. 

We use the information entropy transmitted from the EIH to the robots per cycle to evaluate the communication capability. According to \cite[Theorem 5]{LQR}, in order to achieve a certain LQR cost $l_k$, the average information entropy transmitted through channel $k$ per cycle must satisfy the following constraint
\begin{align}\label{Rl}
	BT_{k}^{\text{commu}}R_k^{\text{U2G}}\left( p_k \right)   \geq \overline{e}_k \left( l_k \right),
\end{align}
where $B$ denotes the bandwidth of each channel, the left side of \eqref{Rl} denotes the maximum information entropy transmitted per cycle, and
\begin{equation}
	\overline{e}_k \left( l_k \right) \triangleq h_k  + \frac{n_k}{2} \log_2 \left( 1+ \frac{ n_k|\det \left( \mathbf{N}_{k} \mathbf{M}_{k}\right) |^\frac{1}{n_k} }{l_k-l_{\text{min},k}} \right)
\end{equation}
denotes the minimum entropy to achieve LQR cost $l_k$, $h_k\triangleq \log_2 |\det \left(  \mathbf{A}_k \right) |$ is the intrinsic entropy rate of object $k$. It should be noted that $h_k$ evaluates the stability of object $k$. An object with a larger intrinsic entropy rate is more difficult to stabilize\cite{control1}. The term $l_{\text{min},k} = \text{tr} \left( \mathbf{\Sigma}^\text{V}_k \mathbf{S}_k\right) +  \text{tr} \left(\mathbf{\Sigma}_k \mathbf{A}_k^\text{T} \mathbf{M}_k \mathbf{A}_k \right) $ denotes the lower bound of the LQR cost, where $\mathbf{N}_{k}$, $\mathbf{M}_{k}$, $\mathbf{\Sigma}_k$ and $\mathbf{S}_k$ are the solutions to the matrix equations shown in \cite{LQR}, which are related to the control parameters, i.e., $\mathbf{A}_k$, $\mathbf{B}_k$, $\mathbf{C}_k$, $\mathbf{R}_k$, $\mathbf{Q}_k$, $\mathbf{\Sigma}^{\text{V}}_k$, and $\mathbf{\Sigma}^{\text{W}}_k$. 

\begin{remark}
	The constraint in \eqref{Rl} is derived under the assumption that both the sensing noise and control process noise follow Gaussian distributions. This assumption is based on the central limit theorem, as the noise can be regarded as the cumulative outcome of a large number of independent variables. If the noise is non-Gaussian, the expression in \eqref{Rl} will be imprecise, offering only an approximation of the relationship between the information entropy and LQR cost. It is worth noting that if the sensing capability is sufficiently strong so that the sensing noise is negligible, then \eqref{Rl} can be extended to scenarios where the control noise is non-Gaussian, as delineated in \cite[Theorem 1]{LQR}. The revised expression of $\overline{e}_k$ for cases with non-Gaussian control process noise and negligible sensing noise is\cite{LQR}
\begin{equation}
	\overline{e}'_k \left( l_k \right) = h_k  + \frac{n_k}{2} \log_2 \left( 1+ \frac{ n_kN \left( \mathbf{v}_k\right)|\det  \left( \mathbf{M}_{k}\right) |^\frac{1}{n_k} }{l_k-\text{tr} \left( \mathbf{\Sigma}^\text{V}_k \mathbf{S}_k\right) } \right),
\end{equation}
where $N\!\left( \mathbf{v}_k\right) = \frac{1}{2\pi}\exp \left( \frac{2}{n_k} h\!\left(  \mathbf{v}_k \right)  \right) $, and $h\!\left(  \mathbf{v}_k \right)$ is the differential entropy of the noise $\mathbf{v}_k $.  The more complicated case with non-Gaussian noises can be a future study.
\end{remark}

\subsection{Problem Formulation}
In this work, we aim to minimize the sum LQR cost of the $\textbf{SC}^3$ loops by jointly optimizing the transmit power allocation $\mathbf{p}= \left\{p_k\right\}$, the computing capability allocation $\mathbf{f} = \left\{f_{k,i}\right\}$, the satellite-backhaul rate allocation $\mathbf{R} = \left\{R_{k,j}\right\}$, and the data splitting vector $\textbf{D} = \left\{D_{k,r}\right\}$ (where $i\in \left\{1,2\right\}$, $j\in \left\{2,3\right\}$ and $r \in \left\{1,2,3\right\}$), while keeping the information entropy constraint satisfied. The optimization problem is formulated as
\begin{subequations}\label{P1}
	\begin{align}
	\min_{\mathbf{p},\mathbf{f}, \mathbf{R}, \mathbf{D}, \mathbf{l}} \quad&\sum_{k = 1}^{K} l_k  \label{P1a} \\ 
	\mbox{\textit{s.t.}}\quad & \sum_{k = 1}^{K} p_k \leq P_{\text{max}},\label{P1b} 
	\\&D_{k,1}\!+\!D_{k,2}\!+\!D_{k,3} = D_k, \quad  k = 1,2,\cdots,K,\label{P1c}
	\\&\sum_{k = 1}^{K} \left( f_{k,1} + f_{k,2}\right)  \leq F_{\text{max}},\label{P1d}
	\\	&\sum_{k = 1}^{K} \left( R_{k,2} + R_{k,3}\right)  \leq R_{\text{max}}^{\text{U2S}},\label{P1e}
	\\&T_{k}^{\text{comp}}+T_{k}^{\text{commu}}\le T_k,\quad k = 1,2,\cdots,K,\label{P1f}
	\\&BT_{k}^{\text{commu}}R_k^{\text{U2G}}\left( p_k \right)   \geq \overline{e}_k \left( l_k \right),  \label{P1g}
	\end{align}
\end{subequations}
where $\mathbf{l}=\left[ l_1, l_2, \cdots, l_K \right] $. The problem in \eqref{P1} is non-convex due to the non-convex and non-continuous expression of $T_{k}^{\text{comp}}$ as in \eqref{Tcompk}. Therefore, it is difficult to solve \eqref{P1} directly. Next, we will recast \eqref{P1} to an equivalent and more tractable form by decoupling the optimization of data splitting vector $\mathbf{D}$.

\section{Problem Transformation and Iterative Solution}
\label{sec_algorithm}
In this section, we will derive the optimal data splitting vector, given the computing capability and satellite-backhaul rate allocated to each $\textbf{SC}^3$ loop, so that we can decouple the optimization of $\mathbf{D}$, and find a way to simplify the original problem. Then, based on the simplified optimization problem, we propose an iterative algorithm to obtain a sub-optimal solution.

\subsection{Problem Transformation}
It can be seen from \eqref{P1} that the LQR cost $l_k$ is influenced by the other variables only through the information entropy constraint \eqref{P1g}. As $\overline{e}_k \left( l_k \right)$ is decreasing with $l_k$, we can prove that $l_k$ is increasing with the computation time $T_{k}^{\text{comp}}$. Therefore, minimizing the LQR cost $l_k$ in $\textbf{SC}^3$ loop $k$ indicates minimizing the computation time $T_{k}^{\text{comp}}$. In addition, it is observed that the data splitting parameters in different $\textbf{SC}^3$ loops are decoupled with each other. Therefore, if the computing and communication resources allocated to $\textbf{SC}^3$ loop $k$ is given, the optimal data splitting vector in that loop can be calculated by minimizing the computation time. Based on the above analysis, we can transform the optimization problem \eqref{P1} to
\begin{subequations}\label{P2}
	\begin{align}
	\min_{\mathbf{p},\mathbf{f}', \mathbf{R}', \mathbf{l}} \ &\sum_{k = 1}^{K} l_k  \label{P2a} \\ 
	\mbox{\textit{s.t.}}\  & \sum_{k = 1}^{K} p_k \leq P_{\text{max}},\label{P2b} 
	\\&\sum_{k = 1}^{K} f_k \leq F_{\text{max}},
	\\&\sum_{k = 1}^{K} R_k \leq R_{\text{max}}^{\text{U2S}},
	\\&T_{k}^{\text{comp},*}\left( f_k, R_k \right)  +T_{k}^{\text{commu}}\le T_k,\  k = 1,2,\cdots,K,\label{P2e}
	\\&BT_{k}^{\text{commu}}R_k^{\text{U2G}}\left( p_k \right)   \geq \overline{e}_k \left( l_k \right),  \label{P2d}
	\end{align}
\end{subequations}
where $f_k = f_{k,1} + f_{k,2}$ denotes the overall computing capability allocated to $\textbf{SC}^3$ loop $k$, $R_k =  R_{k,2} + R_{k,3}$ denotes the satellite-backhaul rate allocated to $\textbf{SC}^3$ loop $k$, and $\mathbf{f}' = \left\{f_k\right\}$ and $\mathbf{R}' = \left\{R_k\right\}$. The function $T_{k}^{\text{comp},*}\left( f_k, R_k \right)$ denotes the minimal computation time of loop $k$ when $f_k$ and $R_k$ are given. $T_{k}^{\text{comp},*}\left( f_k, R_k \right)$ can be calculated as the optimal objective function value of the following optimization problem
\begin{subequations}\label{P3}
	\begin{align}
	\min_{\mathbf{D}_k, \mathbf{f}_k, \mathbf{R}_k} \quad&\max \left\{T_{k,1}^{\text{comp}}, T_{k,2}^{\text{comp}}, T_{k,3}^{\text{comp}}\right\}  \label{P3a} \\ 
	\mbox{\textit{s.t.}}\quad & D_{k,1}\!+\!D_{k,2}\!+\!D_{k,3} = D_k,\label{P3b} 
	\\&f_{k,1}+f_{k,2} \leq f_k,\label{P3c}
	\\&R_{k,2}+R_{k,3} \leq R_k,\label{P3d}
	\end{align}
\end{subequations}
where $\mathbf{D}_k = \left\{D_{k,1}, D_{k,2}, D_{k,3}\right\}$, $\mathbf{f}_k = \left\{f_{k,1}, f_{k,2}\right\}$, and $\mathbf{R}_k = \left\{R_{k,2}, R_{k,3}\right\}$ denote the split schemes in $\textbf{SC}^3$ loop $k$. The optimization problem \eqref{P3} is non-convex due to the non-convex and non-continuous objective function. Next, we will solve the problem \eqref{P3} and give a closed-form expression of $T_{k}^{\text{comp},*}\left( f_k, R_k \right)$.

\subsection{Optimal Solution to Problem \eqref{P3}}
The problem in \eqref{P3} is not convex and therefore difficult to solve directly. In order to solve \eqref{P3}, we have the following lemma.
\begin{lemma}
	\label{lemma1}
	If $f_{k}<\frac{\alpha D_k}{4\tau}$, then the equations 
	\begin{align}
	&T_{k,1}^{\text{comp}} = T_{k,2}^{\text{comp}} = T_{k,3}^{\text{comp}},\label{lemma1_1}\\
	&T_{k,2}^{\text{proc}} = T_{k,2}^{\text{trans}}\label{lemma1_2}
	\end{align}
	must hold in order to minimize the computation time in $\textbf{SC}^3$ loop $k$.
\end{lemma} 

\begin{proof}
	From \eqref{Tk1}, \eqref{Tk2} and \eqref{Tk3}, we can see that the computation time of each part of sensor data is strictly increasing with the respective data size, i.e., $D_{k,1}$, $D_{k,2}$ and $D_{k,3}$. Therefore, if \eqref{lemma1_1} does not hold and $T_{k,i}$ ($i \in \left\{1, 2, 3\right\}$) is larger than the other two terms, we have $T_{k}^{\text{comp}} = T_{k,i}^{\text{comp}}$. We can reduce the corresponding data size $D_{k,i}$ and increase the data size of the other two parts, until \eqref{lemma1_1} holds. Following the above procedure, we decrease  $T_{k,i}^{\text{comp}}$ and increase the computation time of the other two parts, and thereby decreasing the overall computation time $T_k^{\text{comp}}$. Therefore, the equation \eqref{lemma1_1} must hold to minimize the overall computation time if $f_{k}<\frac{\alpha D_k}{4\tau}$ (the condition $f_{k}<\frac{\alpha D_k}{4\tau}$ guarantees that $T_{k}^{\text{comp},*}>4\tau$ and increasing the data size of the second or third parts at the jumping point, i.e., the zero pint, will not increase the overall computation time).
	
	Next, we prove that \eqref{lemma1_2} should hold to minimize $T_{k}^{\text{comp}}$. If  $T_{k,2}^{\text{comp}}>T_{k,2}^{\text{trans}}$, which indicates that the communication resource for the second part of sensor data in loop $k$ is redundant, we can decrease $R_{k,2}$ and increase $R_{k,3}$ until \eqref{lemma1_2} holds. The above procedure will decrease $T_{k,3}^{\text{comp}}$ while maintaining $T_{k,2}^{\text{comp}}$ unchanged, resulting in a non-increasing overall computation time $T_{k}^{\text{comp}}$. On the other hand, if $T_{k,2}^{\text{comp}}<T_{k,2}^{\text{trans}}$, we can decrease $f_{k,2}$ and increase $f_{k,1}$ in a similar way to ensure that $T_{k}^{\text{comp}}$ is non-increasing.
\end{proof}

Based on \textbf{Lemma \ref{lemma1}}, we have the following proposition.

\begin{proposition}
	\label{prop1}
	The optimal value of the objective function of \eqref{P3} is given by the piece-wise function shown in \eqref{Tcompk_fk} on the next page.
	\begin{figure*}[!t]
		\begin{small}
	\begin{equation}\label{Tcompk_fk}
	\!\!\!\!T_{k}^{\text{comp},*}\!\!\left( f_k, R_k \right) \!=\! \left\{
	\begin{aligned}
	&T_{k}^{1}\left( f_k, R_k \right)\!\triangleq\! \frac{\beta D_k}{ \beta R_k \!+\! \left( 1 \!-\! \rho  \right) f_k}\!+\! 4\tau,\quad  \left( f_k,R_k \right) \in \mathcal{S}_1\triangleq \left\{\left( f,R\right) |\  0\leq f\leq\min \left\{\frac{\left( \alpha-\alpha\rho -\beta \right) D_k - 4\beta\tau R}{4\left( 1-\rho \right) \tau},  \frac{\beta R}{\rho} \right\}\right\}\\
	&T_{k}^{2}\left( f_k, R_k \right)\!\triangleq\!\frac{\rho \alpha D_k \!-\! 4\rho \tau f_k \!+\! 4\beta R_k \tau}{ \rho f_k + \left( \alpha - \beta \right) R_k} \!\!+\! 4\tau, \  \left( f_k,R_k \right) \in \mathcal{S}_2 \triangleq \left\{\left( f,R\right) |\  \frac{\beta R}{\rho}\leq f\leq \frac{\left( \alpha-\alpha \rho -\beta \right) D_k - 4\beta \tau R}{4\left( 1-\rho \right) \tau}, R\geq 0\right\}\\
	&T_{k}^{3}\left( f_k, R_k \right)\!\triangleq\!\frac{\alpha D_k-4\tau f_k}{ f_k + \alpha R_k}+4\tau, \  \left( f_k,R_k \right) \in \mathcal{S}_3\triangleq \left\{\left( f,R\right) |\  \frac{\left( \alpha-\alpha \rho -\beta \right) D_k - 4\beta \tau R}{4\left( 1-\rho \right) \tau}\leq f \leq \frac{\alpha D_k}{4\tau},f\geq 0, R\geq 0\right\}\\
	&T_{k}^{4}\left( f_k, R_k \right)\!\triangleq\!\frac{\alpha D_k}{f_k},\quad \left( f_k,R_k \right) \in \mathcal{S}_4\triangleq \left\{\left( f,R\right) |\  f \geq \frac{\alpha D_k}{4\tau},R\geq 0\right\}
	\end{aligned}
	\right
	.
	\end{equation}
	\end{small}
	\hrulefill
	\end{figure*}
\end{proposition}

\begin{proof}
	See Appendix \ref{Appendix1}.
\end{proof}

\begin{remark}
	In most cases, all the computing and communication resources allocated to $\textbf{SC}^3$ loop $k$ will be utilized, and $T_{k}^{\text{comp},*}$ is strictly decreasing with respect to $f_k$ and $R_k$. The exception is when $f_k \geq \frac{\alpha D_k}{4\tau}$, which indicates that the MEC computing capability is enough and the computation time in the MEC server is less than the satellite propagation delay. In such case, all the sensor data will be processed in the EIH, and the satellite communication resource will not be utilized even if $R_k$ is large.
\end{remark}

\begin{remark}
	In the case when $\alpha-\alpha\rho -\beta <0$, we have $\mathcal{S}_1 = \mathcal{S}_2 = \emptyset$. In such case, we have  $D_{k,2} = 0$, i.e., none of the sensor data will be pre-processed in the MEC server. In fact, the condition $m-m\rho -n <0$ holds if $n$ is close to $m$ or $\rho$ is close to $1$, which implies that pre-processing the sensing is not so useful. 
\end{remark}

\subsection{Joint Communication and Computing Resource Optimization}
Based on \textbf{Proposition \ref{prop1}}, we can remove the optimization variables $\mathbf{D}$ and recast \eqref{P1} to an equivalent form as \eqref{P2}. However, \eqref{P2} is still a non-convex optimization problem due to the non-convexity of $T_{k}^{\text{comp},*}\left( f_k, R_k \right) $. Next, we propose an iterative algorithm to solve the joint communication and computing resource optimization problem 

First, we regard the communication time $\left\{ {T}^{\text{commu}}_{k}\right\}$ as optimization variables, and rewrite \eqref{P2} as
\begin{subequations}\label{P4}
	\begin{align}
	\min_{\mathbf{p},\mathbf{f}', \mathbf{R}', \mathbf{l}, \mathbf{T}^{\text{commu}}} \ &\sum_{k = 1}^{K} l_k  \label{P4a} \\ 
	\mbox{\textit{s.t.}}\  & \sum_{k = 1}^{K} p_k \leq P_{\text{max}},\label{P4b} 
	\\&\sum_{k = 1}^{K} f_k \leq F_{\text{max}},\label{P4c}
	\\&\sum_{k = 1}^{K} R_k \leq R_{\text{max}}^{\text{U2S}},\label{P4d}
	\\ \begin{split}\label{P4e}&T_{k}^{\text{comp},*}\left( f_k, R_k \right)  +T_{k}^{\text{commu}}\le T_k,
	\\&\quad\quad\quad\quad\quad\quad\quad\quad\quad k = 1,2,\cdots,K,
	\end{split}
	\\&BR_k^{\text{U2G}}\left( p_k \right)  \geq \frac{\overline{e}_k \left( l_k \right)}{T_{k}^{\text{commu}}}, \quad  k = 1, 2, \cdots, K, \label{P4f}
	\end{align}
\end{subequations}
where $\mathbf{T}^{\text{commu}} = \left\{{T}^{\text{commu}}_{k}\right\}$. By regarding ${T}^{\text{commu}}_{k}$ as a variable and moving it to the left of \eqref{P4f}, we can clarify the convexity of \eqref{P4f}, based on the following lemma.

\begin{lemma}
	\label{lemma2}
	The function $f\left( x, y\right)  = \frac{1}{y}\log\left( 1+\frac{1}{x-a} \right) $ with $a\in \mathbb{R}^+$ is convex in the domain $\textbf{dom} f = \left\{ \left( x, y\right) | x>a, y>0 \right\}$.
\end{lemma}

\begin{proof}
	See Appendix \ref{appendix2}.
\end{proof}

Based on \textbf{Lemma \ref{lemma2}}, it can be shown that the right side of \eqref{P4f}, i.e., $\overline{e}_k \left( l_k \right)/T_{k}^{\text{commu}}$, is convex with respect to $l_k$ and $T_{k}^{\text{commu}}$. In addition, we have $R_k^{\text{U2G}}\left( p_k \right)$ is concave with respect to $p_k$. Therefore, constraint \eqref{P4f} describes a convex set. However, the function $T_{k}^{\text{comp},*}\!\left( f_k, R_k \right)$ in \eqref{P4e} is piece-wise and non-convex, which makes this problem still difficult to solve. Next, we propose an iterative algorithm to solve \eqref{P4} based on the SCA method~\cite{SCA}. Before proceeding further, we introduce the following lemma in order to approximate the non-convex function $T_{k}^{\text{comp},*}\!\left( f_k, R_k \right)$.

\begin{lemma}\label{lemma3}
	For the convex function $1/xy$ with $x>0$ and $y>0$, we have the following inequality for any $x_0>0$ and $y_0>0$
	\begin{align}\label{lemma3_1}
		\frac{1}{xy}&\geq \frac{1}{x_0 y_0}\left( 3-\frac{x}{x_0} -\frac{y}{y_0}\right).
	\end{align}
\end{lemma}

\begin{proof}
The convexity of $1/xy$ can be checked by checking the Hessian matrix. With the convexity of $1/xy$, we can obtain the inequality in \eqref{lemma3_1} immediately through the first-order condition of convex functions ~\cite[Section 3.1.3]{cvx}.
\end{proof}

Based on \textbf{Lemma \ref{lemma3}}, we have the following inequality by substituting $x = 1/u, y = au+bv, x_0 = 1/u_0, y_0 = au_0+bv_0$ into \eqref{lemma3_1}
\begin{equation}\label{lemma3_2}
\frac{u}{au+bv}\geq \frac{u_0}{au_0+bv_0}\left( 3-\frac{u_0}{u}-\frac{au+bv}{au_0+bv_0}\right),
\end{equation}
where $a>0$, $b>0$, $u>0$ and $v>0$.

With \eqref{lemma3_2}, we are ready to approximate $T_{k}^{\text{comp},*}\!\left( f_k, R_k \right)$ with a convex function. The approximate function $\overline{T}_{k}^{\text{comp},*}\left( f_k, R_k |f_{k0},R_{k0}\right)$ is formulated as \eqref{Tk_approx} on the next page, where $f_{k0}>0$ and $R_{k0}>0$ are fixed values, and $\overline{T}_{k}^{2}\left( f_k, R_k | f_{k0}, R_{k0}\right)$ and $\overline{T}_{k}^{3}\left( f_k, R_k | f_{k0}, R_{k0}\right)$ are two convex functions that approximate $T_{k}^{2}\left( f_k, R_k \right)$ and $T_{k}^{3}\left( f_k, R_k \right)$, respectively, formulated as \eqref{Tk2_hat} and \eqref{Tk3_hat} on the next page. The inequalities in \eqref{Tk2_approx} and \eqref{Tk3_approx} follow from \eqref{lemma3_2}, indicating that
\begin{align}
	&\overline{T}_{k}^{2}\left( f_k, R_k | f_{k0}, R_{k0}\right)\geq T_{k}^{2}\left( f_k, R_k \right),\label{Tk2_approx1}
	\\&	\overline{T}_{k}^{3}\left( f_k, R_k | f_{k0}, R_{k0}\right)\geq T_{k}^{3}\left( f_k, R_k \right).\label{Tk3_approx1}
\end{align}

\begin{figure*}[!t]
	\begin{equation}\label{Tk_approx}
		\overline{T}_{k}^{\text{comp},*}\left( f_k, R_k |f_{k0},R_{k0}\right) = \left\{
		\begin{aligned}
			&\max \left\{T_{k}^{1}\left( f_k, R_k \right), \overline{T}_{k}^{2}\left( f_k, R_k | f_{k0}, R_{k0}\right)\right\},\quad  \left( f_{k0},R_{k0} \right) \in \mathcal{S}_1 \cup\mathcal{S}_2\\
			&\overline{T}_{k}^{3}\left( f_k, R_k | f_{k0}, R_{k0}\right),\quad  \left( f_{k0},R_{k0} \right) \in \mathcal{S}_3\\
			&T_{k}^{4}\left( f_k,R_k \right),\quad \left( f_{k0},R_{k0} \right) \in \mathcal{S}_4
		\end{aligned}
		\right		.
	\end{equation}
	\hrulefill
\end{figure*}

\begin{figure*}[!t]
\begin{align}
\overline{T}_{k}^{2}\left( f_k, R_k | f_{k0}, R_{k0}\right)&\triangleq \frac{\rho \alpha D_k}{ \rho f_k + \left( \alpha - \beta \right) R_k}+4\frac{\alpha\tau}{\alpha-\beta}-\frac{\frac{4\rho \alpha\tau}{\alpha-\beta}f_{k0}}{\rho f_{k0} + \left( \alpha - \beta \right) R_{k0}}\left[3 - \frac{f_{k0}}{f_k}-\frac{ \rho f_k + \left( \alpha - \beta \right) R_k}{ \rho f_{k0} + \left( \alpha - \beta \right) R_{k0}}\right] \label{Tk2_hat}\\
&\geq \frac{\rho \alpha D_k}{ \rho f_k + \left( \alpha - \beta \right) R_k}+4\frac{\alpha\tau}{\alpha-\beta}-\frac{\frac{4\rho \alpha \tau}{\alpha -\beta}f_{k}}{\rho f_{k} + \left( \alpha - \beta \right) R_{k}} \label{Tk2_approx}\\
\overline{T}_{k}^{3}\left( f_k, R_k | f_{k0}, R_{k0}\right)&\triangleq\frac{\alpha D_k}{ f_k + \alpha R_k}+4\tau-\frac{4\tau f_{k0}}{f_{k0}+\alpha R_{k0}}\left[ 3-\frac{f_{k0}}{f_k}-\frac{f_k + \alpha R_k}{f_{k0} + \alpha R_{k0}}\right] \label{Tk3_hat}\\
&\geq \frac{\alpha D_k}{ f_k + \alpha R_k}+4\tau-\frac{4\tau f_{k}}{f_{k}+\alpha R_{k}}\label{Tk3_approx}
	\end{align}
		\hrulefill
\end{figure*}

It should be noted that the approximate function $\overline{T}_{k}^{\text{comp},*}\left( f_k, R_k |f_{k0},R_{k0}\right) $ is not a piece-wise function. Instead, the specific expression of $\overline{T}_{k}^{\text{comp},*}\left( f_k, R_k |f_{k0},R_{k0}\right) $ depends on the values of $f_{k0}$ and $R_{k0}$. We have the following lemma which illustrates the fundamental properties of $\overline{T}_{k}^{\text{comp},*}$.

\begin{lemma}
	\label{lemma4}
	The function $\overline{T}_{k}^{\text{comp},*}\left( f_k, R_k |f_{k0},R_{k0}\right) $ shown in \eqref{Tk_approx} is a convex function, and satisfies the following inequality
	\begin{equation}\label{lemma4_1}
	\overline{T}_{k}^{\text{comp},*}\left( f_k, R_k |\ f_{k0},R_{k0}\right)  \geq 	T_{k}^{\text{comp},*}\left( f_k, R_k \right),
	\end{equation}
	where $f_{k0}$ and $R_{k0}$ are non-negative constants, and the equality holds if $f_k = f_{k0}$, and $R_k = R_{k0}$. 
\end{lemma}

\begin{proof}
	See Appendix \ref{appendix3}.
\end{proof}

By approximating $T_{k}^{\text{comp},*}\!\!\left( f_k, R_k \right)$ with $\overline{T}_{k}^{\text{comp},*}\left( f_k, R_k |f_{k0},R_{k0}\right) $, we propose an iterative algorithm to obtain a sub-optimal solution to problem \eqref{P4}.  During each iteration, we solve an approximate optimization problem of \eqref{P4}, formulated as
\begin{subequations}\label{P5}
	\begin{align}		\small
	\min_{\mathbf{p},\mathbf{f}', \mathbf{R}', \mathbf{l}, \mathbf{T}^{\text{commu}}} \ &\sum_{k = 1}^{K} l_k  \label{P5a} \\ 
		\mbox{\textit{s.t.}}\  & \sum_{k = 1}^{K} p_k \leq P_{\text{max}},\label{54b} 
	\\&\sum_{k = 1}^{K} f_k \leq F_{\text{max}},\label{P5c}
	\\&\sum_{k = 1}^{K} R_k \leq R_{\text{max}}^{\text{U2S}},\label{P5d}\\
 \begin{split}\label{P5e}	& \overline{T}_{k}^{\text{comp},*}\!\left( f_k, R_k |f^{\left( i-1\right) }_{k},R^{\left( i-1\right) }_{k}\right) \! + \! T_{k}^{\text{commu}}\! \le \! T_k,
	\\&\quad\quad\quad\quad\quad\quad\quad\quad\quad\quad\quad k = 1,2,\cdots,K,
	\end{split}\\
	&BR_k^{\text{U2G}}\left( p_k \right)  \geq \frac{\overline{e}_k \left( l_k \right)}{T_{k}^{\text{commu}}}, \quad  k = 1, 2, \cdots, K, \label{P5f}
	\end{align}
\end{subequations}
where $i$ denotes the iteration index, and $f^{\left( i-1\right) }_{k}$ and $R^{\left( i-1\right) }_{k}$ denote the solutions in the $\left( i-1\right) $-th iteration. As $\overline{T}_{k}^{\text{comp},*}\left( f_k, R_k |f_{k0},R_{k0}\right)$ is a convex function, it can be proven that problem \eqref{P5} is a convex optimization problem, which can be solved efficiently with convex optimization toolboxes\cite{cvx}.

\begin{algorithm}[t]\label{Algo1}
	\caption{The proposed iterative algorithm for solving problem \eqref{P4}}
	\SetKwInOut{Input}{Input}\SetKwInOut{Output}{Output}\SetKwInOut{Initialize}{Initialization}
	\Input {System parameter $P_{\text{max}}$, $F_{\text{max}}$, $R_{\text{max}}^{\text{U2S}}$, etc; the convergence tolerance $\epsilon$.}
	\Initialize {Calculate a feasible $\mathbf{f}'^0$ and $\mathbf{R}'^0$ based on (\ref{P5c}) and (\ref{P5d}), and set $i = 0$}
	\Repeat{$ \frac{L^{i-1}-L^{i}}{L^{i-1}} <\epsilon$}{
		Set $i = i+1$\;
		Update $\mathbf{p}^i$, $\mathbf{f}'^i$ and $\mathbf{R}'^i$ by solving \eqref{P5}, denote the value of the objective function as $L^i$\; }
	\Output  {the optimal resource allocation $\mathbf{p}^i$, $\mathbf{f}'^i$, $\mathbf{R}'^i$, and the sum LQR cost $L^i$.}
\end{algorithm}

By solving the optimization problem in \eqref{P5} iteratively, we propose \textbf{Algorithm \ref{Algo1}} to solve \eqref{P4}. The convergence of this algorithm can be demonstrated with the following proposition.
\begin{proposition}\label{prop2}
	The output solution of \textbf{Algorithm \ref{Algo1}} is a feasible solution to the optimization problem in \eqref{P4}. In addition, $L^i$ in \textbf{Algorithm \ref{Algo1}} is non-increasing along with the iterations, i.e., $L^{i-1}\geq L^i$ holds for any $i> 1$. Therefore, \textbf{Algorithm \ref{Algo1}} is assured to converge.
\end{proposition}

\begin{proof}
	See Appendix \ref{AppendixD}.
\end{proof}

The computational complexity of the proposed \textbf{Algorithm \ref{Algo1}} to solve \eqref{P4} is dominated by the process of solving \eqref{P5} during the iterations. As \eqref{P5} is a convex optimization problem, it can be solved with interior point method~\cite{cvx}. The complexity of the interior point method is $\mathcal{O}\left( K^{3.5} \log  \left( 1/\epsilon_0 \right) \right) $, where $\epsilon_0$ is the solution accuracy of the interior point method\cite{complexity}. Denoting the iteration number of \textbf{Algorithm \ref{Algo1}} as $I_1$, the overall computational complexity of the proposed algorithm is $ \mathcal{O}\left(I_1 K^{3.5} \log  \left( 1/\epsilon_0 \right) \right)$. In the next section, we will evaluate the iteration numbers via simulations.

\section{Simulation Results and Discussion}
\label{sec_simulation}
In this section, we provide simulation results to evaluate our proposed algorithm. We consider an EIH-empowered $\textbf{SC}^3$ system where the EIH assists $K = 5$ robots for their control tasks. The locations of the robots are assumed to be randomly distributed in a circular area with a radius of $5000$ m. The UAV is located in the center of the circle, with the height of $100$ m. The bandwidth of each channel is set as $B = 5$kHz, and other parameters are set as $\beta_0 = -60$ dB and $\sigma^2  =-110$ dBm~\cite{Hua2018}. We assume a Low Earth Orbit (LEO) satellite, with the height of $1500$ km, and hence $\tau = 5$ ms. The power constraint and the satellite-backhaul rate constraint are set as $P_{\text{max}} = 10$ dBW and $R_{\text{max}}^{\text{U2S}} = 50$ Mbps unless specified otherwise. For the computing parameters, we set $\alpha = 100$ CPU cycles/bit, $\beta = 50$ CPU cycles/bit, and $D_k = 300$ kilobits. Unless specified otherwise, the maximal CPU frequency of the MEC server is set as $F_{\text{max}} = 5$ GHz, and the data compression ratio is set as $\rho = 0.2$.

For control parameters, unless specified otherwise, the state matrices $\mathbf{A}_{k}$ are assumed to be $50\times 50$ diagonal matrices with diagonal elements randomly selected in $\left[ -10, 10\right] $. The control system noise and sensing noise are assumed to be independent Gaussian random variables with zero means and covariance matrices $\mathbf{\Sigma}^{\text{V}}_k = \sigma^2_{\text{V},k}\times\mathbf{I}_n$ and $\mathbf{\Sigma}^{\text{W}}_k = \sigma^2_{\text{W},k}\times\mathbf{I}_n$, where $n = 50$,  $\sigma^2_{\text{V},k} =\sigma^2_{\text{V}} =  0.01$ and $\sigma^2_{\text{W},k} = \sigma^2_{\text{W}} = 0.001$ unless specified otherwise. The time constraint of each loop is set as $70$ ms. The observation matrices are set as identity matrices, and the LQR weight matrices are $\mathbf{Q}_k =  \mathbf{I}_n, \mathbf{R}_k = \mathbf{0}$.

All the simulations are implemented in MATLAB R2021b, and the convex optimization problem  is solved with the fmincon function of the Optimization Toolbox\cite{fmincon}. The interior-point algorithm is used for the fmincon function to solve the convex optimization problems, and the optimality tolerance parameter of the fmincon function is set as $10^{-8}$. The  convergence tolerance threshold is set ot be $\epsilon = 5\times 10^{-5}$. As a suitable initial point, we distribute the computing capability and satellite-backhaul rate equally, i.e., $\mathbf{f}'^0 = \left[ F_{\text{max}}/K, \cdots, F_{\text{max}}/K\right]$, and $\mathbf{R}'^0 = \left[ R_{\text{max}}/K, \cdots, R_{\text{max}}/K\right]$

In order to evaluate the performance of the proposed algorithm, we compare it with the following benchmarks through simulation.
\begin{itemize}
	\item Closed-loop-oriented power allocation: allocating the EIH transmit power to robots, aiming to minimize the sum LQR cost as in \cite{wcl}, where the computing capability and the satellite-backhaul rate are allocated equally to the loops.  
	\item Communication-oriented scheme: allocating the computing capability of MEC server, aiming to minimize the sum computation time\cite{MEC5}, with the satellite-backhaul rate equally allocated to the loops, and the transmit power of the EIH is allocated to maximize the downlink data throughput.
\end{itemize}

Fig. \ref{simu4} verifies the convergence performance of the proposed algorithm. Ten snapshots with different robot locations are evaluated, where the transmit power constraint is set as $P_{\text{max}} = 10$ dBW. This figure shows that our proposed algorithm can converge within three iterations, confirming its efficiency in practical applications.

\begin{figure} [t]
	\centering
	\includegraphics[width=\linewidth]{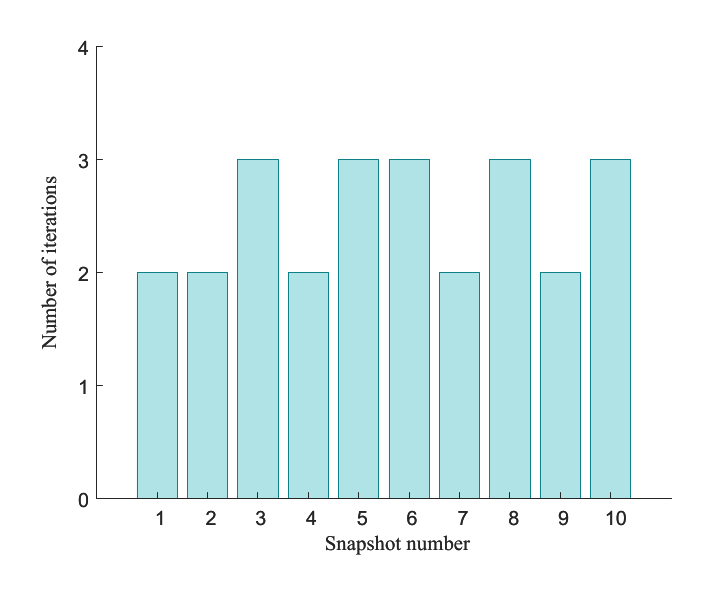}
	\caption{Convergence performance of the proposed scheme.}
	\label{simu4}
\end{figure}

\begin{figure} [t]
	\centering
	\includegraphics[width=\linewidth]{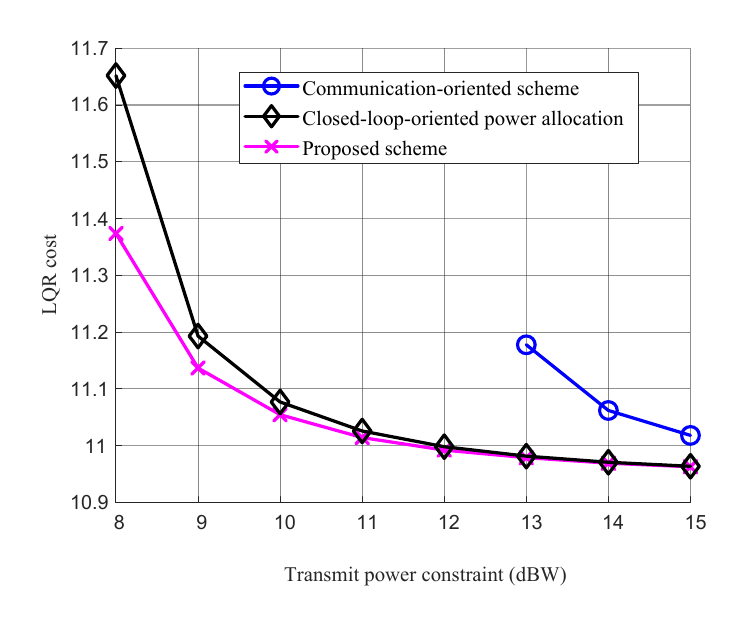}
	\caption{The LQR cost achieved with different transmit power constraints.}
	\label{simu1}
\end{figure}

In Fig. \ref{simu1}, we compare the LQR cost achieved by the above three schemes with different transmit power constraints. From this figure, it is seen that the communication-oriented scheme achieves the worst closed-loop performance. Particularly, the system with the communication-oriented scheme will be unstable when the transmit power constraint is below $12$ dBW, leading to an infinite LQR cost. The proposed scheme achieves the lowest LQR cost. In addition, it is shown that the LQR cost is decreasing with respect to the transmit power constraint, which indicates that improving the communication capability is beneficial for the overall closed-loop performance. However, when the transmit power is sufficiently large, the LQR cost becomes saturated, and the rate of decrease slows.

\begin{figure} [t]
	\centering
	\includegraphics[width=\linewidth]{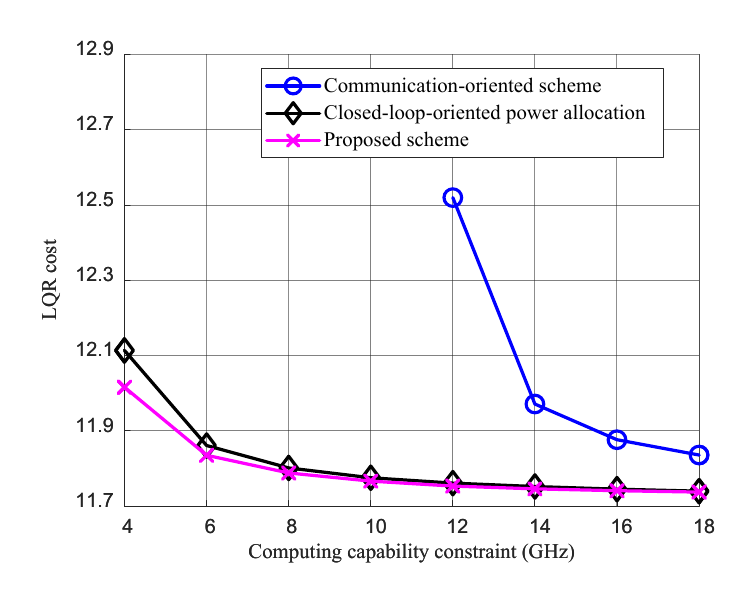}
	\caption{The LQR cost achieved with different computing capability constraints.}
	\label{simu2}
\end{figure}

\begin{figure} [t]
	\centering
	\includegraphics[width=\linewidth]{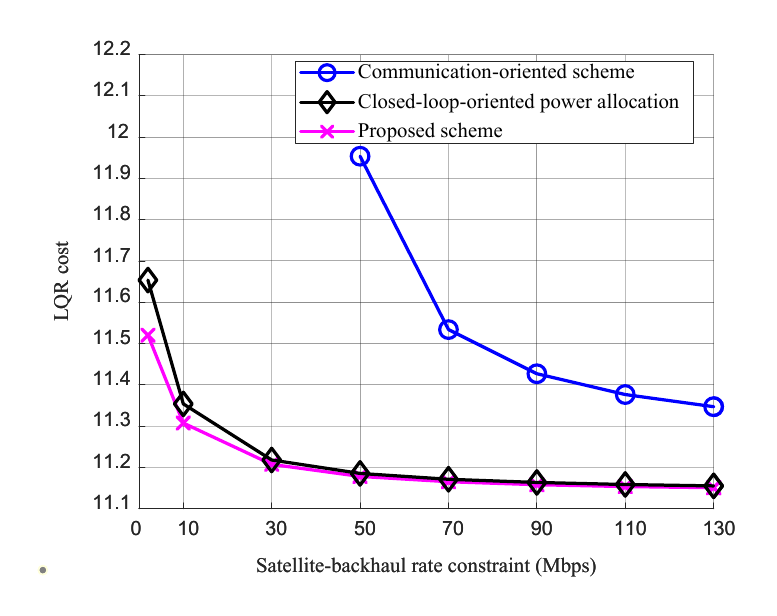}
	\caption{The LQR cost achieved with different satellite-backhaul rate constraints.}
	\label{simu3}
\end{figure}

In Figs. \ref{simu2} and \ref{simu3}, we show the LQR cost with different computing capability constraints $F_{\text{max}}$ and satellite-backhaul rate constraints $R_{\text{max}}$, respectively. We can see that the proposed scheme outperforms the other two schemes under all conditions. Similar to Fig. \ref{simu1}, it is shown that the LQR cost is decreasing with respect to the $F_{\text{max}}$ and $R_{\text{max}}$, indicating that increasing the computing capability can improve the closed-loop control performance of $\textbf{SC}^3$ loops.

\begin{figure} [t]
	\centering
	\includegraphics[width=\linewidth]{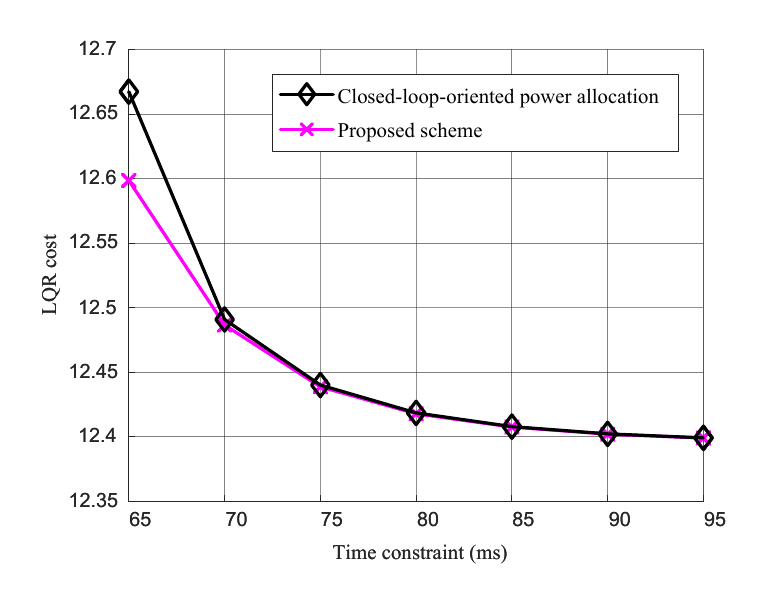}
	\caption{The LQR cost achieved with different time constraints.}
	\label{simu8}
\end{figure}

We show the impact of the time constraints $T_k$ on the LQR cost in Fig. \ref{simu8}, where the time constraints of each $\textbf{SC}^3$ loop are set to be the same for better readability. The maximal CPU frequency is set as $F_{\text{max}}=3$ GHz. We can see that the proposed scheme can achieve a lower LQR cost than the power allocation scheme, especially when the time resources are limited, thereby showing the superiority of the proposed scheme. They become similar when the time resources are saturated, i.e., $T_k$ is larger than $75$ ms. In addition, it can be seen that the LQR cost decreases with the time constraint, which indicates that increasing the time for communication and computing is an efficient way to improve the control performance. The reason is that more precise control commands can be transmitted given more time for transmission and computation.

\begin{figure} [t]
	\centering
	\includegraphics[width=\linewidth]{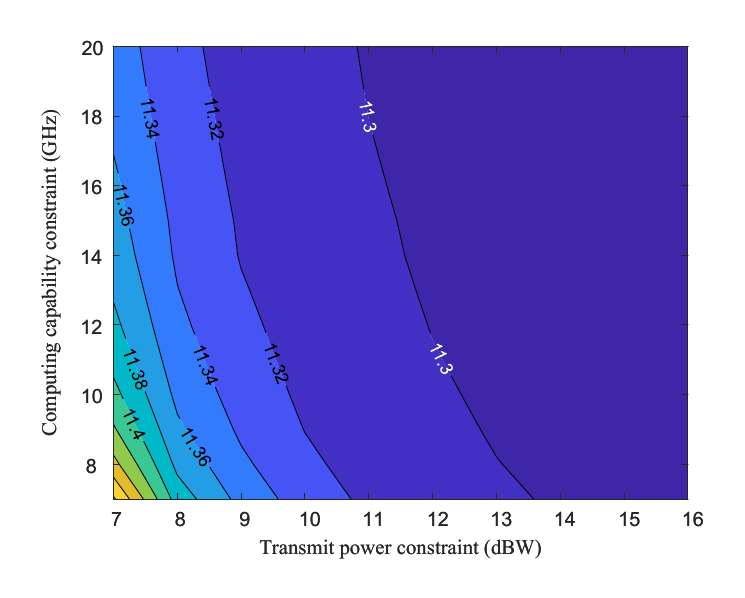}
	\caption{LQR cost achieved with different transmit power and computing capability constraints.}
	\label{simu5}
\end{figure}

Fig. \ref{simu5} shows how the LQR cost is influenced by the transmit power constraints and computing capability constraints with the proposed scheme. It is shown that the LQR cost decreases with both the transmit power and computing capability constraints. However, the contours become sparse as the transmit power constraint or computing capability constraint increases, indicating diminishing marginal returns with respect to transmit power and computing capability. In addition, it can be seen that even for the computing capability with a high value, the LQR cost is still restricted by a lower bound that is determined by the maximum power. The reason is that the transmit time $T_{k}^{\text{commu}}$ cannot increase infinitely with the increased computing capability, leading to an upper bound of the transmitted information entropy. The finite information entropy determines the bound of LQR cost $l_k$, as shown in the constraint \eqref{P1g}.

\begin{figure} [t]
	\centering
	\includegraphics[width=\linewidth]{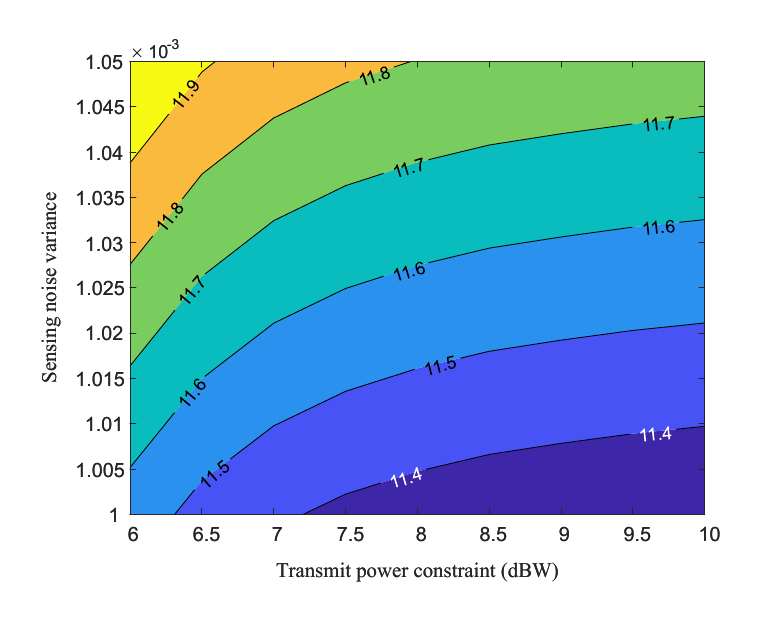}
	\caption{LQR cost achieved with different transmit power constraints and sensing noise variances.}
	\label{simu6}
\end{figure}

To show the joint influence of the sensing and communication capability on the closed-loop performance, we show the contours of the sum LQR cost with respect to the transmit power constraint $P_{\text{max}}$ and the sensing noise variance $\sigma^2_{\text{V}}$ in Fig. \ref{simu6}. It is shown that the increase of the sensing noise variance will cause the degradation of the control performance, leading to a higher LQR cost. This degradation can be compensated partially by enhancing the communication capability, i.e., increasing $P_{\text{max}}$. However, even if the transmit power constraint becomes high enough, the LQR cost will still be bounded by the lower bound $l_{\text{min},k}$, which will also be increased by the sensing part. This result shows that only enhancing the communication capability on the EIH cannot fully compensate for the poor sensing capability in the $\textbf{SC}^3$ loop.

\begin{figure} [t]
	\centering
	\includegraphics[width=\linewidth]{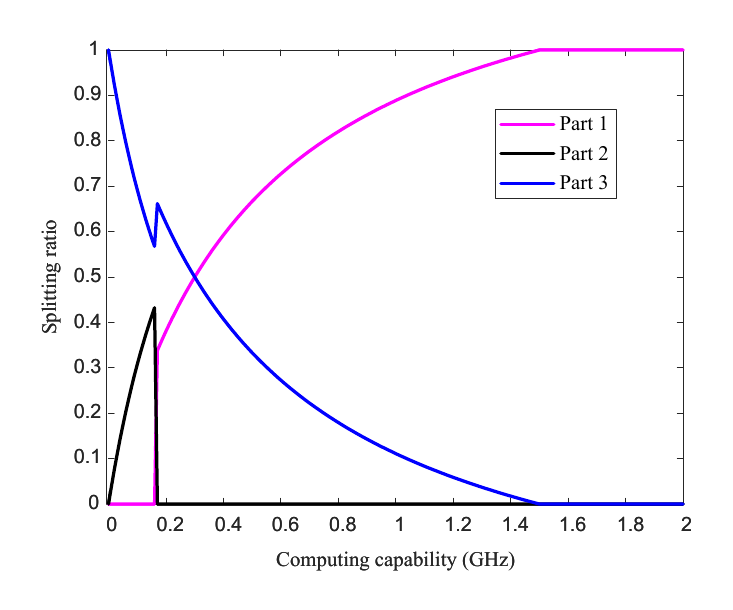}
	\caption{Optimal splitting ratio of the different parts of sensor data with different computing capabilities.}
	\label{simu7}
\end{figure}

Fig. \ref{simu7} shows the optimal data splitting vector of the sensor data in  $\textbf{SC}^3$ loop $k$ with different computing capabilities $f_k$ based on \textbf{Proposition \ref{prop1}}, where $R_k = 50$ Mbps and $\rho = 0.25$. It can be seen that when the local computing capability is low, most sensor data will be transmitted to the cloud server for processing (Part 3), and the local computing capability will be fully used on pre-processing the sensor data (Part 2). As the computing capability increases, the sensor data will be either processed locally in the MEC server (Part 1) or in the cloud server (Part 3). Finally, when the local computing capability on the EIH is large enough, all the sensor data will be processed locally, and we have $D_{k,1} = D_k$, and $D_{k,2} = D_{k,3} = 0$.

\begin{figure} [t]
	\centering
	\includegraphics[width=\linewidth]{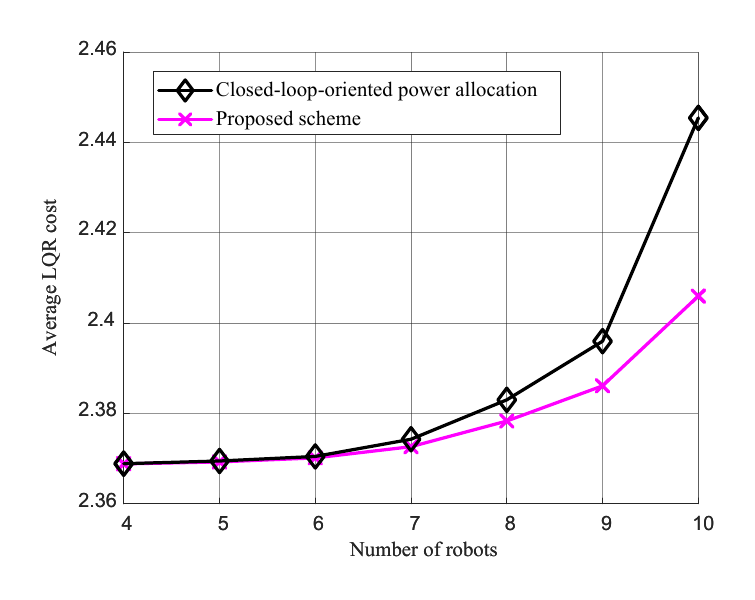}
	\caption{The average LQR cost achieved with different numbers of robots.}
	\label{simu9}
\end{figure}

In Fig. \ref{simu9}, we show the average LQR cost of each $\textbf{SC}^3$ loop with different numbers of robots, i.e., $K$. For ease of comparison, the control parameters of each $\textbf{SC}^3$ loop are set to be identical. It can be seen that the average LQR cost increases rapidly with the number of robots. The reason is that the communication and computing resources allocated to each $\textbf{SC}^3$ loop will decrease as the number of robots increases. This indicates that an efficient allocation of resources is important in practical scenarios where there are many robots working together. It is worth noting that when the number of robots increases, different robots will cooperate to complete a common task, and the coupling between different $\textbf{SC}^3$ loops will be interesting and non-negligible, which will be left to the future study.

\section{Conclusion}
\label{sec_conclusion}
In this paper, we investigated an EIH-empowered satellite-UAV network serving multiple robots for their control tasks. The UAV-mounted EIH is integrated with sensing, computing, and communication modules. It is capable of directing the behaviors of robots, via synergistic $\textbf{SC}^3$ closed-loop orchestration. In order to explore the potential of closed-loop optimization, we formulated a sum LQR cost minimization problem that jointly optimized the splitting of sensor data, the computing capability, the satellite-backhaul rate, and the transmit power from EIH to robots. An iterative algorithm was proposed to solve this non-convex optimization problem. Simulation results demonstrated the superiority of the proposed algorithm. Moreover, we have shown the joint influences of the sensing, communication, and computing capability on the sum LQR cost, to uncover a more systematic understanding of closed-loop controls.

\appendices
\section{Proof of Proposition \ref{prop1}}\label{Appendix1}
	First, if $f_{k}\geq\frac{\alpha D_k}{4\tau}$, all the sensor data should be processed locally in the MEC server of the EIH, and the overall computation time is $T_{k}^{\text{comp}} = T_{k,1}^{\text{comp}} = \frac{\alpha D_k}{f_k}<4\tau$. Otherwise, if some sensor data are transmitted to the cloud server through satellite, i.e., $D_{k,2}>0$ or $D_{k,3}>0$, the computation time will be larger than the sum propagation delay $4\tau$, as shown in \eqref{Tk2} and \eqref{Tk3}.
	
	Next, we consider the case of $f_{k}<\frac{\alpha D_k}{4\tau}$. Based on \eqref{Tk1}, \eqref{Tk2} and \eqref{Tk3}, we can see that the latency is non-increasing with respect to the $f_{k,1}$, $f_{k,2}$, $R_{k,2}$ and $R_{k,3}$. Therefore, the equality of \eqref{P3c} and \eqref{P3d} must hold to minimize the overall computation time, i.e,
	\begin{align}
		&f_{k,1}+f_{k,2} = f_k,\label{Appendix1_1}
		\\&R_{k,2}+R_{k,3} = R_k\label{Appendix1_2}.
	\end{align}
	Therefore, we have six equations in total, as shown in \eqref{Dk}, \eqref{lemma1_1}, \eqref{lemma1_2}, \eqref{Appendix1_1}, and \eqref{Appendix1_2}. Based on these equations, next, we will remove the other variables and express the objective function of \eqref{P3} as a function of $f_{k,2}$. 

	From \eqref{lemma1_2}, \eqref{Appendix1_1}, and \eqref{Appendix1_2} , we have $f_{k,1} = f_k-f_{k,1}$, $R_{k,2} = \frac{\rho}{\beta}f_{k,2}$ and $R_{k,3} = R_k-\frac{\rho}{\beta}f_{k,2}$. Substituting these variables into \eqref{Dk} and \eqref{lemma1_1}, we can get the relation between $T_k^{\text{comp}}$ and $f_{k,2}$, formulated as
\begin{align}\label{Tcompk_fk2}
T_k^{\text{comp}} = \frac{\alpha \beta D_k-4\beta\tau f_k + 4\beta\tau f_{k,2}}{\left( \alpha - \alpha\rho - \beta\right) f_{k,2} + \beta f_k + \alpha \beta R_k}+4\tau.
\end{align}

From \eqref{Tcompk_fk2}, we see that $T_k^{\text{comp}}$ is a fractional linear function of $f_{k,2}$. Therefore, $T_k^{\text{comp}}$ is monotonous with respect to $f_{k,2}$. It can be shown that $T_k^{\text{comp}}$ is increasing with respect to $f_{k,2}$ if $4\beta \tau R_k-\left( \alpha-\alpha\rho -\beta \right) D_k + 4\left( 1-\rho \right) \tau f_k \geq 0$, and $T_k^{\text{comp}}$ is decreasing otherwise.

Based on the above analysis, if $4\beta \tau R_k-\left( \alpha-\alpha\rho -\beta \right) D_k + 4\left( 1-\rho \right) \tau f_k \geq 0$, then $f_{k,2}$ should be as small as possible in order to minimize $T_k^{\text{comp}}$. Therefore, we have $f_{k,2} = 0$ in such case. Substituting $f_{k,2} = 0$ into \eqref{Tcompk_fk2}, we have
\begin{align}\label{Tcompk_fk2_1}
T_k^{\text{comp}} = \frac{\alpha D_k-4\tau f_k}{ f_k + \alpha R_k}+4\tau.
\end{align}

On the other hand, if $4\beta \tau R_k-\left( \alpha-\alpha \rho -\beta \right) D_k + 4\left( 1-\rho \right) \tau f_k \leq 0$, then $f_{k,2}$ should be as large as possible. There are two constraints for $f_{k,2}$, i.e., $f_{k,2}\leq f_k$ and $R_{k,2} = \frac{\rho}{\beta}f_{k,2}\leq R_k$, which indicates that $f_{k,2} = \min\left\{ {f_k, \frac{\beta R_k}{\rho}}\right\}$. Therefore, when $f_k \leq \frac{\beta R}{\rho}$, we have $f_{k,2} = f_{k}$, and
\begin{align}\label{Tcompk_fk2_2}
T_k^{\text{comp}} = \frac{\beta D_k}{ \beta R_k + \left( 1-\rho  \right) f_k}+4\tau.
\end{align}

On the other hand, if $f_k \geq \frac{\beta R}{\rho}$, then $f_{k,2} = \frac{\beta R_k}{\rho}$, $R_{k,2} = R_k $, and
\begin{align}\label{Tcompk_fk2_3}
T_k^{\text{comp}} = \frac{\rho \alpha D_k - 4\rho \tau f_k + 4\beta R_k \tau}{ \rho f_k + \left( \alpha - \beta \right) R_k}+4\tau.
\end{align}

Based on the above analysis, the validity of \textbf{Proposition \ref{prop1}} has been demonstrated.

\section{Proof of Lemma \ref{lemma2}}
\label{appendix2}
	According to \cite[Page 89]{cvx}, a function is convex if and only if its epigraph is a convex set, where the epigraph of $f\left( x, y \right) $ is defined as
$\textbf{epi} f = \left\{ (x,y,z) |(x,y)\in \textbf{dom} f, f\left( x, y\right) \le z \right\}$.

In order to show the convexity of $f\left( x, y \right)$, we will show that its epigraph $\mathcal{A}$ is a convex set, which can be formulated as
\begin{equation}\label{Lemam1_2_a}
\mathcal{A} = \left\{ (x,y,z) |\quad\frac{1}{y}\log\left( 1+\frac{1}{x-a} \right) \le z, x>a, y>0 \right\}.
\end{equation}
We can transform $\mathcal{A}$ to an equivalent form as
\begin{equation}\label{Lemma1_3_b}
\mathcal{A} = \left\{ (x,y,z) |\quad \frac{1}{\exp\left( yz\right) -1 }+a\le x, y>0, z>0 \right\}.
\end{equation}
From \eqref{Lemma1_3_b}, it can be seen that $\mathcal{A}$ can be regarded as the epigraph of a new function $g\left( y, z\right)  = \frac{1}{\exp\left( yz\right) -1 }+a$ with the domain $y>0,z >0$. 

Next, we show that $g\left( y, z\right)$ is convex by checking its Hessian matrix, which is calculated as
\begin{equation}
\nabla ^2 g\left( y, z \right)  = 
\begin{bmatrix}
\frac{z^2 e^{yz} \left( e^{yz} +1 \right) }{\left( e^{yz} -1\right)^3 }& \frac{e^{yz} \left( yz-e^{yz} +yz e^{yz}+1 \right) }{\left( e^{yz} -1\right)^3 }\\
\frac{e^{yz} \left( yz-e^{yz} +yz e^{yz}+1 \right) }{\left( e^{yz} -1\right)^3 }&\frac{y^2 e^{yz} \left( e^{yz} +1 \right) }{\left( e^{yz} -1\right)^3 }
\end{bmatrix}.
\end{equation} 
The determinant of $\nabla ^2 g\left( y, z \right)$ can be calculated as
\begin{align}
|\nabla ^2 g\left( y, z \right)|  =  \frac{e^{2yz} \left( 2yze^{yz}-e^{yz}+2yz+1\right) }{\left( e^{yz} -1\right)^5}.
\end{align}

It can be shown that $2yze^{yz}-e^{yz}+2yz+1>0$ when $y>0,z>0$ by checking the derivative. Therefore, we have $|\nabla ^2 g\left( y, z \right)|>0$. As $\frac{\partial^2 g}{\partial y^2}>0$ and $\frac{\partial^2 g}{\partial z^2}>0$, we can obtain that $\nabla ^2 g\left( y, z \right)$ is positive definite with  $y>0$ and $z>0$, which shows the convexity of $g$. From the above analysis, we have $\mathcal{A}$ is a convex set, indicating that $f\left( x, y \right)$ is a convex function.

\section{Proof of Lemma \ref{lemma4}}
\label{appendix3}
We first prove the convexity of $\overline{T}_{k}^{\text{comp},*}\left( f_k, R_k |f_{k0},R_{k0}\right)$. As $T_{k}^{1}\left( f_k, R_k \right)$, $\overline{T}_{k}^{2}\left( f_k, R_k | f_{k0}, R_{k0}\right)$, $\overline{T}_{k}^{3}\left( f_k, R_k | f_{k0}, R_{k0}\right)$ and $T_{k}^{4}\left( f_k, R_k \right)$ are all reciprocal functions of the linear combination of $f_K$ and $R_k$, they are all convex functions. As the point-wise maximum function of two convex functions is still a convex function\cite{cvx}, we can establish the convexity of $\overline{T}_{k}^{\text{comp},*}\left( f_k, R_k |f_{k0},R_{k0}\right)$.

Next, we prove the correctness of inequality \eqref{lemma4_1} by comparing the values of $T_{k}^{i}\left( f_k, R_k \right)$ for $i \in \left[ 1, 2, 3, 4 \right] $. The differences of the four functions can be formulated as
\begin{align}
&T_{k}^{1}\left( f_k, R_k \right) - T_{k}^{2}\left( f_k, R_k \right)\nonumber\\
= &\frac{\left(\rho f_k \!-\! \beta R_k \right) \left[4\left( 1 \!-\!\rho\right) \tau f_k \!-\! {\left(\alpha \!-\!\alpha\rho  \!-\!\beta \right)D_k\! +\! 4\beta\tau R_k }\right] }{\left[ \beta R_k \!+\! \left( 1 \!-\! \rho  \right) f_k\right]\left[  \rho f_k + \left( \alpha - \beta \right) R_k \right] } ,\label{AppendixC_1_1}\\
&T_{k}^{1}\left( f_k, R_k \right) - T_{k}^{3}\left( f_k, R_k \right) \nonumber\\
= & \frac{f_k\left[4\left( 1 \!-\!\rho\right) \tau f_k \!-\! {\left(\alpha \!-\!\alpha\rho  \!-\!\beta \right)D_k\! +\! 4\beta\tau R_k }\right]}{\left[ \beta R_k \!+\! \left( 1 \!-\! \rho  \right) f_k\right]\left( f_k + \alpha R_k\right) },\label{AppendixC_1_2}\\
&T_{k}^{2}\left( f_k, R_k \right) - T_{k}^{3}\left( f_k, R_k \right)\nonumber\\
= & \frac{\alpha R_k\left[4\left( 1 \!-\!\rho\right) \tau f_k \!-\! {\left(\alpha \!-\!\alpha\rho  \!-\!\beta \right)D_k\! +\! 4\beta \tau R_k }\right]}{\left( f_k + \alpha R_k\right)\left[  \rho f_k + \left( \alpha - \beta \right) R_k \right] },\label{AppendixC_1_3}\\
&T_{k}^{3}\left( f_k, R_k \right) - T_{k}^{4}\left( f_k, R_k \right)  \nonumber
\\=& \frac{\alpha R_k\left(4f_k \tau - \alpha D_k \right) }{f_k \left( f_k+\alpha R_k \right) }.\label{AppendixC_1_4}
\end{align}

Based on the above results, we can establish the relationship among the four functions when $\left( f_k, R_k \right)$ falls in different areas, that is
	\begin{small}
		\begin{subequations}\label{AppendixC_2}
\begin{align}
T_{k}^{4}\left( f_k, R_k \right)>T_{k}^{3}\left( f_k, R_k \right)>T_{k}^{1}\left( f_k, R_k \right),\  \left( f_k, R_k \right)\in \mathcal{S}_1, \label{AppendixC_2_1}\\
T_{k}^{4}\left( f_k, R_k \right)>T_{k}^{3}\left( f_k, R_k \right)\geq T_{k}^{2}\left( f_k, R_k \right),\  \left( f_k, R_k \right)\in \mathcal{S}_2, \label{AppendixC_2_2}\\
T_{k}^{1}\left( f_k, R_k \right)\geq T_{k}^{2}\left( f_k, R_k \right)\geq T_{k}^{3}\left( f_k, R_k \right),\  \left( f_k, R_k \right)\in \mathcal{S}_3, \label{AppendixC_2_3}\\
T_{k}^{4}\left( f_k, R_k \right)\geq T_{k}^{3}\left( f_k, R_k \right),\  \left( f_k, R_k \right)\in \mathcal{S}_3, \label{AppendixC_2_4}\\
T_{k}^{2}\left( f_k, R_k \right)\geq T_{k}^{3}\left( f_k, R_k \right)>T_{k}^{4}\left( f_k, R_k \right),\  \left( f_k, R_k \right)\in \mathcal{S}_4. \label{AppendixC_2_5}
\end{align}
		\end{subequations}
\end{small}

With the inequalities in \eqref{AppendixC_2}, it can be proven that, for any $ f_{k}\geq 0$, and $R_{k} \geq 0 $, the following inequalities hold
\begin{subequations}\label{AppendixC_3}
\begin{align}
T_{k}^{\text{comp},*}	\left( f_k, R_k \right)&\leq \max \left\{T_{k}^{1}\left( f_k, R_k \right), T_{k}^{2}\left( f_k, R_k\right)\right\},\label{AppendixC_3_1}\\
T_{k}^{\text{comp},*}	\left( f_k, R_k \right)&\leq T_{k}^{3}\left( f_k, R_k \right), \\
T_{k}^{\text{comp},*}	\left( f_k, R_k \right)&\leq T_{k}^{4}\left( f_k, R_k \right).
\end{align}
\end{subequations}

Therefore, if $\left( f_{k0},R_{k0} \right) \in \mathcal{S}_1 \cup\mathcal{S}_2$, we have
\begin{subequations}\label{AppendixC_6}
	\begin{align}
&\overline{T}_{k}^{\text{comp},*}\left( f_{k}, R_{k} |\ f_{k0},R_{k0}\right)\label{AppendixC_6_1}
\\ = &\max \left\{T_{k}^{1}\left( f_k, R_k \right), \overline{T}_{k}^{2}\left( f_k, R_k | f_{k0}, R_{k0}\right)\right\}\label{AppendixC_6_2}
\\ \geq & \max \left\{T_{k}^{1}\left( f_k, R_k \right), T_{k}^{2}\left( f_k, R_k\right)\right\}\label{AppendixC_6_3}
\\ \geq & T_{k}^{\text{comp},*}\left( f_k, R_k \right),\label{AppendixC_6_4}
	\end{align}
\end{subequations}
where \eqref{AppendixC_6_2} is based on the definition of $\overline{T}_{k}^{\text{comp},*}$ in \eqref{Tk_approx}, \eqref{AppendixC_6_3} follows from \eqref{Tk2_approx1}, and \eqref{AppendixC_6_4} follows from \eqref{AppendixC_3_1}. 

If $\left( f_{k0},R_{k0} \right) \in \mathcal{S}_3$ or $\left( f_{k0},R_{k0} \right) \in \mathcal{S}_4$, it can be proven that $\overline{T}_{k}^{\text{comp},*}\left( f_k, R_k |\ f_{k0},R_{k0}\right)  \geq 	T_{k}^{\text{comp},*}\left( f_k, R_k \right)$ following a similar procedure as \eqref{AppendixC_6}, which demonstrates the correctness of \eqref{lemma4_1}.

Finally, we show the equality condition of \eqref{lemma4_1}, i.e.,
\begin{equation}\label{AppendixC_4}
\overline{T}_{k}^{\text{comp},*}\left( f_{k0}, R_{k0} |\ f_{k0},R_{k0}\right) = T_{k}^{\text{comp},*}	\left( f_{k0},R_{k0} \right).
\end{equation}
If $\left(  f_{k}, R_{k} \right) \in \mathcal{S}_1$, we have 
\begin{subequations}
	\begin{align}
	&\overline{T}_{k}^{\text{comp},*}\left( f_{k0}, R_{k0} |\ f_{k0},R_{k0}\right)\label{AppendixC_4_1}\\
	= &\max \left\{T_{k}^{1}\left(f_{k0}, R_{k0} \right), \overline{T}_{k}^{2}\left( f_{k0}, R_{k0} | f_{k0}, R_{k0}\right)\right\}\label{AppendixC_4_2}\\
	= &\max \left\{T_{k}^{1}\left( f_{k0}, R_{k0} \right), T_{k}^{2}\left( f_{k0}, R_{k0}\right)\right\}\label{AppendixC_4_3}\\
	= &T_{k}^{1}\left( f_{k0}, R_{k0} \right)\label{AppendixC_4_4}\\
	= & T_{k}^{\text{comp},*}\left( f_{k0}, R_{k0} \right),
	\end{align}
\end{subequations}
where \eqref{AppendixC_4_4} follows from \eqref{AppendixC_1_1}. If $\left(  f_{k}, R_{k} \right)$ falls into other areas, the equality can be proven in a similar way, which completes the proof.

\section{Proof of Proposition \ref{prop2}}\label{AppendixD}
Denoting the optimal solution to problem \eqref{P5} in the $i$-th iteration as $\left( \mathbf{p}^i,\mathbf{f}'^i, \mathbf{R}'^i, \mathbf{l}^i, \mathbf{T}^{\text{commu},i} \right) $, we have
\begin{subequations}
	\begin{align}
	&T_{k}^{\text{comp},*}\left( f^i_k, R^i_k \right) +T_{k}^{\text{commu},i}\label{AppendixD_1_1}\\
	\leq \ &\overline{T}_{k}^{\text{comp},*}\left( f^i_k, R^i_k |f^{\left( i-1\right) }_{k},R^{\left( i-1\right) }_{k}\right) +T_{k}^{\text{commu},i}\label{AppendixD_1_2}\\
	\leq \ &T_k\label{AppendixD_1_3}
	\end{align}
\end{subequations}
holds for $k = 1, 2, \cdots, K$, where \eqref{AppendixD_1_2} follows from \eqref{lemma4_1}, and \eqref{AppendixD_1_3} follows from \eqref{P5e}. Therefore, any optimal solution to \eqref{P5} will also satisfy all the constraints in \eqref{P4}, i.e., \eqref{P4b}-\eqref{P4f}, indicating that it is also a feasible solution to \eqref{P4}.

Next, we show the convergence of \textbf{Algorithm \ref{Algo1}}. We have
\begin{subequations}
	\begin{align}
	 &\overline{T}_{k}^{\text{comp},*}\left( f^{\left( i-1\right) }_{k},R^{\left( i-1\right) }_{k} |f^{\left( i-1\right) }_{k},R^{\left( i-1\right) }_{k}\right) +T_{k}^{\text{commu},\left( i-1\right) }\label{AppendixD_2_1}\\
	 =\ &T_{k}^{\text{comp},*}\left( f^{\left( i-1\right) }_{k},R^{\left( i-1\right) }_{k}\right) +T_{k}^{\text{commu},\left( i-1\right)}\label{AppendixD_2_2}\\
	 \leq &\overline{T}_{k}^{\text{comp},*}\left( f^{\left( i-1\right) }_{k},R^{\left( i-1\right) }_{k} |f^{\left( i-2\right) }_{k},R^{\left( i-2\right) }_{k}\right) +T_{k}^{\text{commu},\left( i-1\right)}\label{AppendixD_2_3}\\
	\leq \ &T_k,\label{AppendixD_2_4}
	\end{align}
\end{subequations}
where \eqref{AppendixD_2_3} follows from \eqref{lemma4_1}, and \eqref{AppendixD_2_4} holds becasue $\mathbf{f}'^{\left( i-1\right)}$, $\mathbf{R}'^{\left( i-1\right)}$, and $\mathbf{T}^{\text{commu},\left( i-1\right)}$ are the optimal solution to \eqref{P5} in the $\left( i-1 \right) $-th iteration and should satisfy the constraint \eqref{P5e}. Therefore, it is shown that $\left( \mathbf{p}^{\left( i-1\right) },\mathbf{f}'^{\left( i-1\right) }, \mathbf{R}'^{\left( i-1\right) }, \mathbf{l}^{\left( i-1\right) }, \mathbf{T}^{\text{commu},{\left( i-1\right) }} \right) $ is also feasible to the optimization problem \eqref{P5} in the $i$-th iteration, indicating that $L^{i-1}$ is also an achievable objective function value in the $i$-th iteration. As $\left( \mathbf{p}^i,\mathbf{f}'^i, \mathbf{R}'^i, \mathbf{l}^i, \mathbf{T}^{\text{commu},i} \right) $ minimizes \eqref{P5} in the $i$-th iteration, we have $L^{i}\leq L^{i-1}$ holds for any $i\geq 1$, which completes the proof.




 
\vspace{11pt}
\begin{IEEEbiography}[{\includegraphics[width=1in,height=1.25in,clip,keepaspectratio]{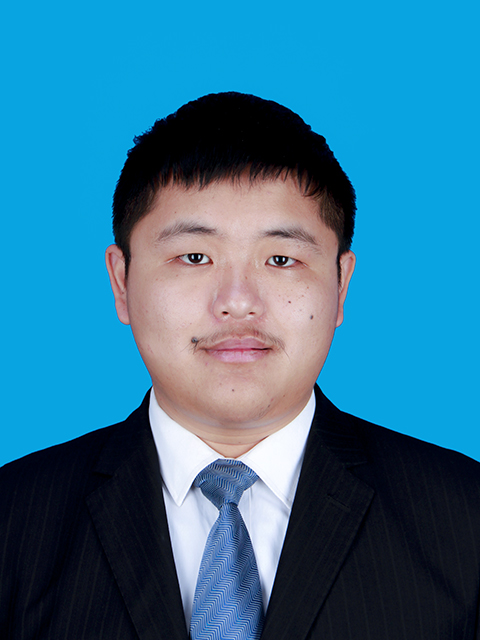}}]{Chengleyang Lei}
	received the B.S. degree in 2021 from the Department of Electronic Engineering, Tsinghua University, Beijing, China. He is currently pursuing the Ph.D degree with the Department of Electronic Engineering in Tsinghua University. His research interests include the communication and control integration, coordinated satellite-UAV-terrestrial networks, and future 6G technologies.
\end{IEEEbiography}

\begin{IEEEbiography}[{\includegraphics[width=1in,height=1.25in,clip,keepaspectratio]{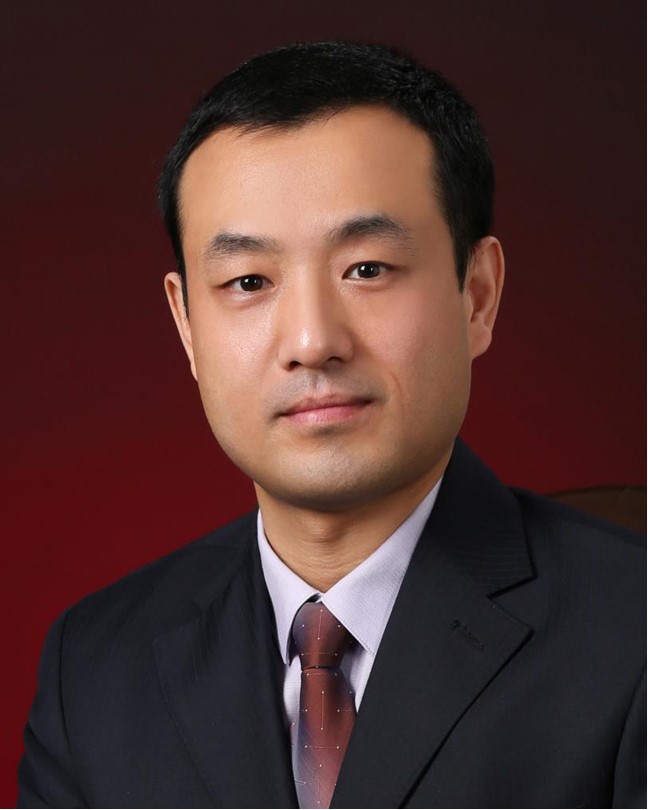}}]{Wei Feng}(Senior Member, IEEE)
	received the B.S. and Ph.D. degrees from the Department of Electronic Engineering, Tsinghua University, Beijing, China, in 2005 and 2010, respectively. 
He is currently a Professor with the Department of Electronic Engineering, Tsinghua University. 
His research interests include maritime communication networks, large-scale distributed antenna systems, and coordinated satellite-UAV-terrestrial networks. 
He serves as the Assistant to the Editor-in-Chief of \textsc{China Communications} and an Editor of \textsc{IEEE Transactions on Cognitive Communications and Networking}.
\end{IEEEbiography}

\begin{IEEEbiography}[{\includegraphics[width=1in,height=1.25in,clip,keepaspectratio]{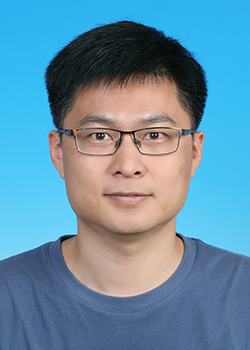}}]{Peng Wei}(Member, IEEE)
received the Ph.D. degree in communication and information systems from the University of Electronic Science and Technology of China in 2017. He is currently an associate Professor with the National Key Laboratory of Wireless Communications, University of Electronic Science and Technology of China. His research interests are in wireless communication, multicarrier system, and signal processing.
\end{IEEEbiography}

\begin{IEEEbiography}[{\includegraphics[width=1in,height=1.25in,clip,keepaspectratio]{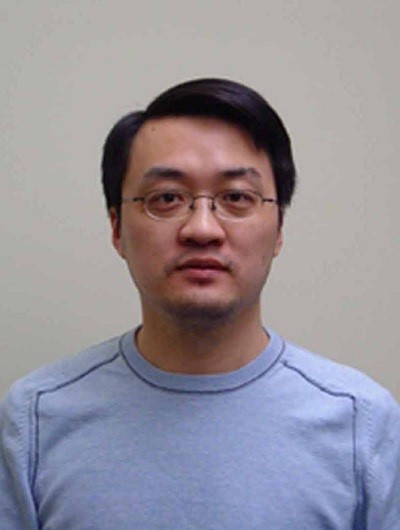}}]{Yunfei Chen}(Senior Member, IEEE)
	received the B.E. and M.E. degrees in electronics engineering from Shanghai Jiaotong University, Shanghai, China, in 1998 and 2001, respectively, and the Ph.D. degree from the University of Alberta in 2006. 
He is currently a Professor with the Department of Engineering, University of Durham, U.K. 
His research interests include wireless communications, cognitive radios, wireless relaying, and energy harvesting.
\end{IEEEbiography}

\begin{IEEEbiography}[{\includegraphics[width=1in,height=1.25in,clip,keepaspectratio]{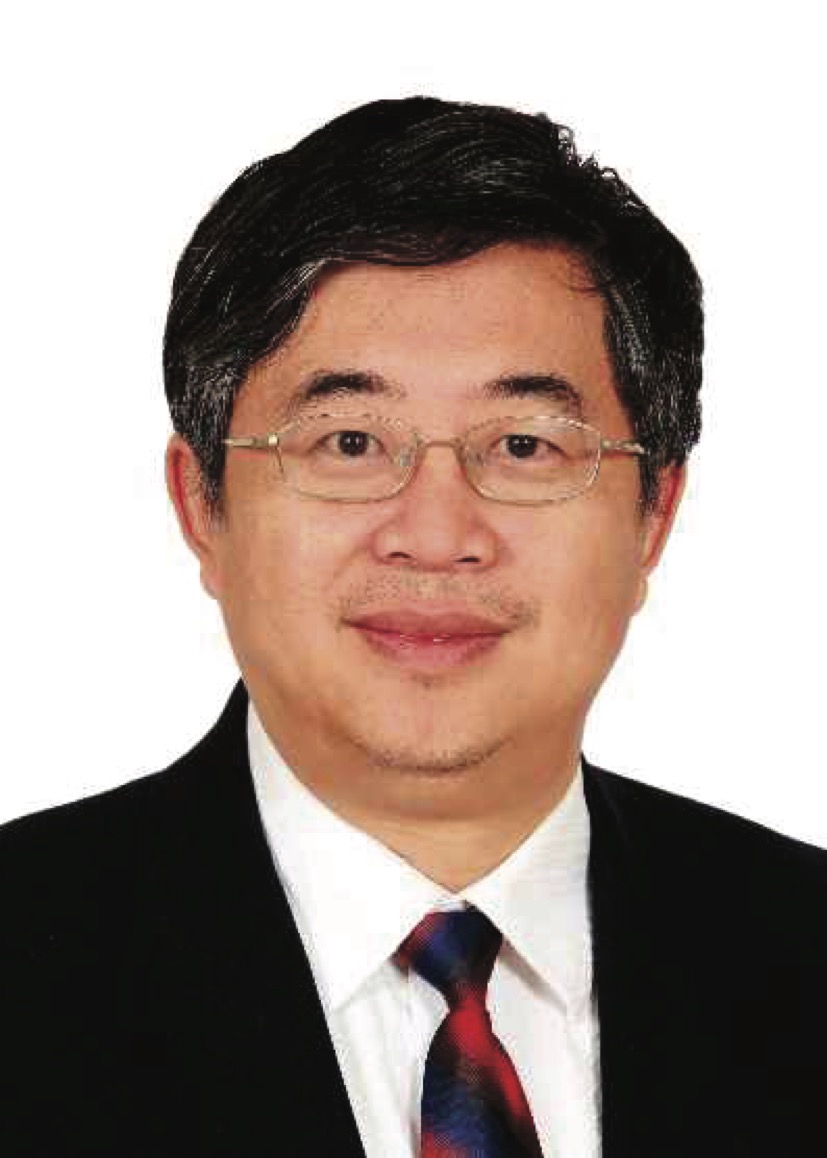}}]{Ning Ge}(Member, IEEE)
	received the B.S. and Ph.D. degrees from Tsinghua University, Beijing, China, in 1993 and 1997, respectively.
From 1998 to 2000, he was with ADC Telecommunications, Dallas, TX, USA, where he researched the development of ATM switch fabric ASIC. 
Since 2000, he has been a Professor with the Department of Electronics Engineering, Tsinghua University. 
He has published over 100 papers. 
His current research interests include communication ASIC design, short range wireless communication, and wireless communications.
Dr. Ge is a senior member of the China Institute of Communications and the Chinese Institute of Electronics.
\end{IEEEbiography}

\begin{IEEEbiography}[{\includegraphics[width=1in,height=1.25in,clip,keepaspectratio]{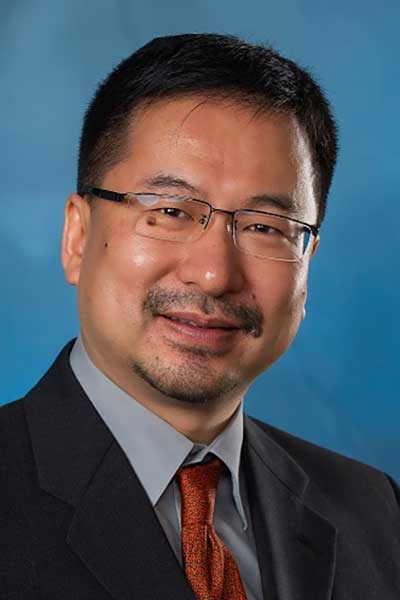}}]{Shiwen Mao}(Fellow, IEEE)
	received his Ph.D. in electrical engineering from Polytechnic University, Brooklyn, NY in 2004.  
Currently, he is a professor and Earle C. Williams Eminent Scholar Chair, and Director of the Wireless Engineering Research and Education Center at Auburn University. 
His research interest includes wireless networks, multimedia communications, and smart grid. 
He is the Editor-in-Chief of \textsc{IEEE Transactions on Cognitive Communications and Networking} and a Distinguished Lecturer of IEEE Communications Society. 
He received the IEEE MMTC Outstanding Researcher Award in 2023, the IEEE TC-CSR Distinguished Technical Achievement Award in 2019 and NSF CAREER Award in 2010. 
He is a co-recipient of the 2022 Best Journal Paper Award of IEEE ComSoc eHealth Technical Committee (TC), the 2021 Best Paper Award of Elsevier/KeAi Digital Communications and Networks Journal, the 2021 IEEE Internet of Things Journal Best Paper Award, the 2021 IEEE Communications Society Outstanding Paper Award, the IEEE Vehicular Technology Society 2020 Jack Neubauer Memorial Award, the 2018 Best Journal Paper Award and the 2017 Best Conference Paper Award from IEEE ComSoc Multimedia TC, and the 2004 IEEE Communications Society Leonard G. Abraham Prize in the Field of Communications Systems. 
He is a co-recipient of the Best Paper Awards from IEEE ICC 2022 and 2013, IEEE GLOBECOM 2023, 2019, 2016, and 2015, and IEEE WCNC 2015, and the Best Demo Awards from IEEE INFOCOM 2022 and IEEE SECON 2017. 
\end{IEEEbiography}

\vspace{11pt}


\vfill

\end{document}